\documentclass[12pt]{article}

\usepackage[margin=2.5cm]{geometry}

\usepackage{amsmath,amsthm,amsfonts,amscd,amssymb,bbm,mathrsfs,enumerate,url}

\usepackage{graphicx,tikz}
\usepackage[pdfborder={0 0 0}]{hyperref}
\usetikzlibrary{arrows}

\def\RR{{\mathbb R}}
\def\CC{{\mathbb C}}
\def\NN{{\mathbb N}}
\def\ZZ{{\mathbb Z}}

\def\A{{\mathcal A}}

\def\D{{\mathcal D}}

\def\H{{\mathcal H}}
\def\I{{\mathcal I}}

\def\M{{\mathcal M}}

\def\P{{\mathcal P}}
\def\R{{\mathcal R}}

\def\a{\alpha}
\def\b{\beta}
\def\d{\delta}

\def\g{\gamma}

\def\k{\kappa}

\def\l{\lambda}

\def\L{{\mathrm L}}

\def\R{{\mathrm R}}

\def\t{\tau}

\def\up{\upsilon}
\def\x{\xi}

\def\Ad{{\hbox{\rm Ad\,}}}

\def\1{{\mathbbm 1}}

\def\u1net{{\A^{(0)}}}

\def\diff{{\rm Diff}}

\def\diffs1{\diff(S^1)}
\def\mob{{\rm M\ddot{o}b}}
\def\mob2{{\rm M\ddot{o}b}^{(2)}}

\def\supp{{\rm supp\,}}

\def\psl2r{{\rm PSL}(2,\RR)}
\def\sl2r{{\rm SL}(2,\RR)}
\def\su11{{\rm SU}(1,1)}
\def\2dmob{{\overline{\psl2r}\times\overline{\psl2r}}}
\def\<{\langle}
\def\>{\rangle}

\def\re{\mathrm{Re}\,}
\def\im{\mathrm{Im}\,}

\def\poincare{{\P^\uparrow_+}}

\def\dom{{\mathrm{Dom}}}
\def\fct{\widetilde{\phi}}
\def\strip{\RR+i(0,\pi)}

\newcommand{\Om}{\Omega}

\DeclareMathOperator{\lspan}{span}

\newtheorem{theorem}{Theorem}[section]

\newtheorem{proposition}[theorem]{Proposition}

\theoremstyle{remark}

\title{Wedge-local fields in integrable models with bound states II. Diagonal S-matrix}
\date{}
\author{
{\bf Daniela Cadamuro} \\
e-mail: {\tt daniela.cadamuro@mathematik.uni-goettingen.de}\\
Mathematisches Institut, Universit\"at G\"ottingen \\
Bunsenstrasse 3-5, D-37073 G\"ottingen, Germany.  \\
{\bf Yoh Tanimoto}\\
e-mail: {\tt hoyt@ms.u-tokyo.ac.jp}\\
Graduate School of Mathematical Sciences, The University of Tokyo\\
3-8-1 Komaba Meguro-ku Tokyo 153-8914, Japan.\\
JSPS SPD postdoctoral fellow\\
}
\begin{document}
\maketitle
\begin{abstract}
We construct candidates for observables in wedge-shaped regions for a class of
$1+1$-dimensional integrable quantum field theories
with bound states whose S-matrix is diagonal, by extending our previous methods for scalar S-matrices.
Examples include the $Z(N)$-Ising models, the $A_{N-1}$-affine Toda field theories and some S-matrices with CDD factors.

We show that these candidate operators which are associated with elementary particles commute weakly on a dense domain.
For the models with two species of particles, we can take a larger domain of weak commutativity
and give an argument for the Reeh-Schlieder property.
\end{abstract}

\section{Introduction}\label{introduction}
Important developments in the construction of $1+1$ dimensional
quantum field theories with factorizing S-matrices have been obtained in recent years in the operator-algebraic approach.
In the class of scalar S-matrices with no poles in the physical strip, Lechner \cite{Lechner03, Lechner08}
considered a large family of analytic functions which satisfy certain regularity conditions
and constructed quantum field theories which have these functions as the two-particle S-matrix.
The idea is to construct first observables localized in an infinitely extended wedge-shaped region.
Lechner and Sch\"utzenhofer \cite{LS14} generalized this construction of wedge-local observables to theories with
a richer particle spectrum, which include among others the $O(N)$-invariant nonlinear $\sigma$-models.
More recently, Alazzawi and Lechner (see \cite{Alazzawi14}) showed that the existence of strictly local observables
in this class of models should follow if certain representations of the symmetric groups $\mathfrak{S}_n$
can be nicely intertwined.
Candidates for observables in wedges in certain massless models have also been found and their relation
to conformal field theories have been investigated by one of the authors \cite{DT11, Tanimoto12-2, BT13, LST13, Tanimoto14-1, BT15}.

In \cite{CT15-1}, we further generalized Lechner's construction to scalar models with S-matrices which have poles
in the physical strip, which are believed to correspond to the presence of bound states. In these models, the scalar
two-particle S-matrix has only one pair of poles in the physical strip,
and it is interpreted that two bosons of the same species fuse
into another boson of the same species (the Bullough-Dodd model is believed to have such properties).
For this class of S-matrices, we constructed operators $\fct(f) := \phi(f) + \chi(f)$ by adding the
bound-state operator $\chi(f)$ to the field $\phi(f)$ of Lechner, and we showed that $\fct(f)$ weakly commute
with its reflected operator $\fct'(g)$ on a common domain.
Hence this $\fct(f)$ is a good candidate for a wedge-local observable.

In this work, we extend this construction to models with a richer particle spectrum.
The models which can be treated with our methods have diagonal S-matrix and
include the $Z(N)$-Ising model and the $A_{N-1}$-affine Toda field theories
and other S-matrices with a CDD factor.
These models are characterized by $N-1$ species of particles and
an intricate pole structure in the physical strip, which include simple and double poles.
Moreover, the fusion process between the $k$-th species and the $l$-th species results in the $(k+l \mod N)$-th species,
hence the models realize a more realistic binding process (c.f.\! the Bullough-Dodd model we studied previously,
where the bound state of the same two species is again the same species).

They have already been studied before in the \textit{form factor programme} \cite{KS79, BK03, BFK06}.
However, the convergence of the expansion of $n$-point functions in terms of form factors remains an open problem
and Wightman fields for these models are not available today.
Our goal is to attain a realization of these models in the operator-algebraic framework, i.e. the Haag-Kastler axioms.

In this framework, we construct candidates for wedge-local observables by extending the construction carried out
in \cite{CT15-1} to a multi-particle situation. Our candidate operator $\fct(f)$ is given by the multi-particle component
field $\phi(f)$ of Lechner-Sch\"utzenhofer \cite{LS14} by adding a multi-particle extension of the bound state operator 
$\chi(f)$ introduced in \cite{CT15-1}. We also construct the reflected operator $\fct'(g)$ and show that
those components corresponding to ``elementary particles'' commute weakly on a dense domain.
In the $Z(N)$-Ising models, these elementary particles are those which have index $1$ or $N-1$.
They are again polarization-free generators (PFGs, an operator localized in a wedge which generates
a one-particle state from the vacuum) \cite{BBS01} but not temperate.

The question of strong commutativity remains open. Since this time it is impossible to choose $f$ such that $\chi(f)$ is positive in these models,
the mere existence of a self-adjoint extension is more complicated. We think that this domain issue of unbounded operators is
an essential feature of the models with bound states and deserves a separate study (c.f.\! \cite{Tanimoto15-1}).
We will see that even the domain of weak commutativity is subtler in general, but
the models with two species of particles, including the $Z(3)$-Ising model and the $A_2$-affine Toda field theory, behave better.
For these subclass of models, we also argue that the Reeh-Schlieder property holds, once the issue of strong commutativity is settled.

The paper is organized as follows. In Sec.\! \ref{zf}, we introduce our general notation for multi-particle Fock space
recalling the results of Lechner-Sch\"utzenhofer \cite{LS14}. Moreover, we list the general properties of
diagonal S-matrices
with poles in the physical strip, of which examples are the $Z(N)$-Ising models and the $A_{N-1}$-affine Toda field theories
and those with an extra (CDD) factor.
Sec.\! \ref{chi} is dedicated to the construction of the bound state operators $\chi(f)$, $\chi'(g)$.
We specify their domains and show symmetry properties as quadratic forms.
In Sec.\! \ref{wedge-local} we construct candidate operators $\fct(f)$ and $\fct'(g)$ and prove the weak
wedge-commutativity between the components for ``elementary particles''.
Moreover, we argue that the $Z(3)$-Ising model and the $A_2$-affine Toda field theory satisfy
the Reeh-Schlieder property.
We expect that this holds in general.
In Sec.\! \ref{sec:conclusions}, we summarize our results and discuss open problems.

\section{Zamolodchikov-Faddeev algebra}\label{zf}

Our construction of wedge-local observables is an extension of \cite{CT15-1} to a larger class of integrable models with
a richer particle spectrum. For a general overview of the program of constructing Haag-Kastler nets from
given S-matrices, see \cite[Section 2.1]{CT15-1}. We consider quantum field theories in $1+1$ dimensional Minkowski space
with factorizing S-matrices, which are characterized by a matrix-valued two-particle scattering function
fulfilling a number of properties. These properties, without poles in the physical strip, have already been treated
in the operator-algebraic framework by Lechner-Sch\"utzenhofer \cite{LS14} and Alazzawi \cite{Alazzawi14}.
In the following, we recall the mathematical framework and the notation we will use to describe these models,
following \cite{LS14}. 

\subsection{Diagonal S-matrix}\label{S-matrix}

In the model with sharp mass shells, we have the well-defined concept of one-particle states.
By isolating the irreducible representations of the Poincar\'e group, we assign them indices $\a, \b,\dots.$
As our models have factorizing S-matrix, their two-particle scattering process can be specified
by the matrix-valued function $S^{\a\b}_{\g\d}(\theta)$, where $\theta$ is the difference of rapidities of incoming particles.

In this work, we restrict ourselves to models in which two incoming particles of types $\alpha, \beta$
result in two outgoing particles of types $\beta, \alpha$.
Such an S-matrix is called ``diagonal'' and has only non-zero components $S^{\a\b}_{\b\a}(\theta)$.
We list in Section \ref{examples} examples of such diagonal S-matrices and refer to
literature for their Lagrangian description, if any.

As in the scalar case \cite{CT15-1}, our model can also include fusion processes.
Indeed, each pole of the component $S^{\a\b}_{\b\a}$ (of the so-called $s$-channel) in the physical strip corresponds to
such a fusion process, and one needs to incorporate them in order to keep locality of the model.
For simplicity, in this paper we assume that the fusion of two species is just one species.
More complicated processes appear e.g.\! in the sine-Gordon model, which we will study in a separate work \cite{CT-sine}.

\subsubsection*{The scattering data}
Now the S-matrix $S$ is specified by the following data:
\begin{itemize}
 \item \textbf{index set} $\I$, which is a finite set, $|\I| = K$.
 The elements of this set are denoted by Greek letters, such as $\a, \b, \g, \d, \mu, \nu, \up$.
 Each element corresponds to a single species of particle in the model.
 \item \textbf{charge structure}: for each index $\a$, there is the conjugate charge $\bar \a \in \I$, which is another index.
 It holds that $\bar {\bar \a} = \a$.
 \item \textbf{masses} $\{m_\a\}$: for each index $\a$, there is a positive number $m_\a > 0$, which is the mass
 of the particle of the species $\a$. We consider only massive particles.
 \item \textbf{fusion table}: to some pairs of indices $\a, \b$, there corresponds another index $\g, \a\neq\g\neq \b$.
 There is a list of all such correspondences and we call it the fusion table of the model.
 When we write $(\a\b) \to \g$, there is a correspondence from the pair $\a,\b$ to $\g$ and
 this is called a \textbf{fusion (process)}.
 In this paper, we assume that for a pair $\a, \b$, there is only one index $\g$ such that $(\a\b)\to\g$
 \footnote{We are aware that this is not the general case. We plan to investigate the sine-Gordon model where
 two breathers can fuse into different two breathers \cite{Quella99} in a separate paper \cite{CT15-3}}.
 
 We assume further that if $(\alpha\beta) \rightarrow \gamma$ is an entry of the fusion table, so are
 $(\beta \alpha) \rightarrow \gamma$, $(\g\bar \a) \to \b$, $(\g\bar\b)\to\a$ and $(\bar \a \bar \b) \rightarrow \bar\g$.
 The entries of the fusion tables are called an \textbf{$s$-channel} of the fusion (e.g.\! \cite{Quella99, Korff:2000}).
 \item \textbf{fusion angles} $\{\theta_{(\a\b)}\}$: if $(\a\b) \to \g$ is a fusion process, then there is
 a positive number $\theta_{(\a\b)} \in (0,\pi)$ and there holds:
 \begin{equation}\label{prelations}
p_{m_{\alpha}}(\zeta +i\theta_{(\alpha \beta)}) +p_{m_{\beta}}(\zeta -i\theta_{(\beta \alpha)}) = p_{m_{\gamma}}(\zeta),
 \end{equation}
 where $p_m(\zeta) = \left(\begin{array}{c} m \cosh\zeta \\ m \sinh\zeta \end{array}\right)$.
 By putting $\zeta = 0$ and considering only the real part, we can depict this relation as Figure \ref{fig:mass}.
 Eq~\eqref{prelations} resembles energy-momentum conservation, but not quite due to complex arguments. One could interpret this relation as the situation where two ``virtual'' particles fuse into a third ``real'' particle (the bounded particle) whose momentum lies in the mass shell \cite[Section 1.3.7]{Quella99}.

 \item \textbf{S-matrix components} $\{S^{\alpha \beta}_{\beta \alpha}\}$: for each pair $\a,\b$ of indices there is a meromorphic function
 $S^{\alpha \beta}_{\beta \alpha}(\zeta)$ on $\CC$  with the following properties (c.f.\! \cite[Definition 2.1]{LS14}) for $\zeta \in \mathbb{C}$:
\begin{enumerate}
  \renewcommand{\theenumi}{(S\arabic{enumi})}
  \renewcommand{\labelenumi}{\theenumi}
 \item \label{unitarity} {\bf Unitarity.} $S^{\alpha \beta}_{\beta \alpha}(\zeta)^{-1} = \overline{S^{\b\a}_{\a\b}(\bar\zeta)}$.
 \item \label{parity} {\bf Parity symmetry.} $S^{\alpha \beta}_{\beta \alpha}(\zeta) = S^{\beta \alpha}_{\alpha \beta}(\zeta)$.
 \item \label{hermitian} {\bf Hermitian analyticity.} $S^{\alpha \beta}_{\beta \alpha}(\zeta) = S^{\beta \alpha}_{\alpha \beta}(-\zeta)^{-1}$.
 \item \label{crossing} {\bf Crossing symmetry.} $S^{\alpha \beta}_{\beta \alpha}(i\pi -\zeta) = S^{\bar\beta \alpha}_{\alpha \bar\beta}(\zeta)$.
 \item \label{CPTinvariance} {\bf CPT invariance.} $S^{\alpha \beta}_{\beta \alpha}(\zeta) = S^{\bar\alpha \bar\beta}_{\bar\beta \bar\alpha}(\zeta)$.
 \item \label{bootstrap} {\bf Bootstrap equation.} If $(\alpha \beta) \rightarrow \gamma$ is a fusion process, there holds
 \begin{equation}\label{eq:bootstrap}
 S^{\gamma \nu}_{\nu \gamma}(\zeta) = S^{\alpha \nu}_{\nu \alpha}(\zeta +i\theta_{(\alpha \beta)}) S^{\beta \nu}_{\nu \beta}(\zeta -i\theta_{(\beta \alpha)}).
 \end{equation}
 \item \label{polestructure} {\bf Pole structure.} For each fusion $(\alpha \beta) \rightarrow \gamma$, $S^{\alpha \beta}_{\beta \alpha}(\zeta)$ has a simple pole at $\zeta = i\theta_{\alpha \beta}$, $0< \theta_{\alpha \beta} < \pi$, where
 \begin{align*}
 \theta_{\alpha \beta} := \theta_{(\alpha \beta)} + \theta_{(\beta \alpha)}.
 \end{align*}
 Furthermore, $S^{\a\b}_{\b\a}$ has another simple pole at $\zeta = i\theta'_{\beta\bar\alpha} := i\pi -i\theta_{\beta \bar\alpha}$ if
 and only if $(\beta \bar\alpha)$ is also a fusion. This is consistent with crossing symmetry.
 If $(\a\b) \to \g$ is a fusion process, then the pole in $S_{\bar\a\b}^{\b\bar\a}$ at $\zeta = i\theta'_{\alpha\beta}$
 is called the \textbf{$t$-channel pole}. In the physical terminology,
 we insert only fusion processes corresponding to $s$-channel poles in the fusion table.
 
 For a fusion process $(\a\b)\to\g$, we denote
 \begin{equation}\label{def:residues}
 R_{\alpha \beta}^\g := \operatorname{Res}_{\zeta = i\theta_{\alpha \beta}} S^{\alpha \beta}_{\beta \alpha}(\zeta),
 \quad R^{\prime\g}_{\alpha \beta} := \operatorname{Res}_{\zeta =i\theta'_{\alpha \beta}} S_{\bar\a\b}^{\b\bar\a}(\zeta),
 \end{equation}
 In general, we define $R_{\a\b}^\g = R_{\a\b}^{\prime\g} = 0$ if $(\a\b)\to\g$ is not a fusion process.
 
 \item \label{atzero} {\bf Value at zero.} $S^{\a\a}_{\a\a}(0) = -1$.
 \item \label{regularity} \textbf{Regularity.} Each component of $S$ has only finitely many zeros and
 there is $\epsilon > 0$ such that $\|S\|_\epsilon := \sup\{|S^{\a\b}_{\b\a}(\zeta)|: \zeta \in \RR + i(-\epsilon, \epsilon), \a,\b\in\I\} < \infty$.
\end{enumerate}

\item \textbf{Elementary particles}. 
There is a distinguished index $\up$ such that $\up \neq \bar\up$ called the elementary particle
(we will call $\bar \up$ an elementary particle as well, but we fix an index $\up$)
which has the following properties:
\begin{itemize}
\item $S^{\up\b}_{\b\up}$ has only simple poles (no higher poles) or no pole at all for each $\b$.
\item $S^{\up\b}_{\b\up}$ has at most two simple poles in the physical strip $\RR + i(0,\pi)$,
one corresponding to the fusion process $(\up\b) \to \g$ and the other which is a $t$-channel pole
for the process $(\g\bar\up)\to \b$ (this properties is known as \textbf{maximal analyticity}).
\item The fusion angle $\theta_{(\a\up)}$ does not depend on $\a$ and it holds that
$\theta_{(\a\up)} = \theta_{(\a\bar\up)}$.
\item It holds that $\theta_{(\up\a)} + \theta_{(\bar\up\a)} = \pi$
for any $\a$ which is not an elementary particle.
\item If $(\up\up)\to\k$ is a fusion process, we say that $\k$ is a composite particle.
Recursively, if $(\up\b) \to \g$ where $\b$ is composite,
then $\g$ is again a composite particle. We assume that each index in $\I$ is either elementary or composite.
In other words, one can arrive at any index by making successive fusions by an elementary particle. 
\item Let $\k$ be the index such that $(\up\up)\to\k$.
Then $\k$ is the unique index for which $S_{\up\k}^{\k\up}(\zeta)$ has a pole in $\RR + i[0,\theta_{(\k\up)}]$.
 \item \label{positive-residue} {\bf Positive residue.} For each fusion process $(\a\up)\to\g$
 including $\up$,
 it holds that $R_{\a\up}^\g \in i\RR_+$. As we remarked in \cite{CT15-1}, this property is related with
 the unitarity of Hamiltonian (if it exists at all).
 Note that, from other properties, it is automatic that $R_{\a\up}^\g$ is purely imaginary
 but the condition $R_{\a\up}^\g \in  i\RR_+$ does not follow and it is crucial for our main results.
\end{itemize}

\end{itemize}
Note that \ref{regularity} refers to the supremum in a neighborhood of the real line and not
of the whole physical strip, because $S$ now has poles in our cases.
In the cases without poles \cite{LS14, Alazzawi14}, the condition \ref{regularity} has been used only in the (attempt at a) proof of modular nuclearity.
We need it already here, when we apply the Cauchy theorem, as we will see.

Note that the fusion angles $\theta_{(\alpha \beta)}$ depend only on the mass ratios of the particles involved
(see Figure \ref{fig:mass}).
\begin{figure}[ht]
\centering
\begin{tikzpicture}[line cap=round,line join=round,>=triangle 45,x=1.0cm,y=1.0cm]
\clip(-3.37,-0.74) rectangle (4.51,4.76);
\draw [domain=-3.37:4.51] plot(\x,{(--7.88-0*\x)/3.94});
\draw (0.5,-0.74) -- (0.5,4.76);
\draw(0.5,2) circle (1.61cm);
\draw(0.5,2) circle (2.58cm);
\draw [->,line width=1.2pt] (0.5,2) -- (1.18,3.46);
\draw [->,line width=1.2pt] (0.5,2) -- (3.08,2);
\draw(0.5,2) circle (2.4cm);
\draw [->,line width=1.2pt] (0.5,2) -- (2.4,0.54);
\draw [->] (2.4,0.54) -- (3.08,2);
\draw [->] (1.18,3.46) -- (3.08,2);
\draw (0.88,3) node[anchor=north west] {$\scriptstyle m_\alpha$};
\draw (1.1,1.25) node[anchor=north west] {$\scriptstyle m_\beta$};
\draw (1.5,2.44) node[anchor=north west] {$\scriptstyle m_\gamma$};
\draw [shift={(0.5,2)}] plot[domain=0:1.13,variable=\t]({1*0.34*cos(\t r)+0*0.34*sin(\t r)},{0*0.34*cos(\t r)+1*0.34*sin(\t r)});
\draw (0.68,2.53) node[anchor=north west] {$\scriptstyle \theta_{(\alpha\beta)}$};
\draw [shift={(0.5,2)}] plot[domain=-0.66:0,variable=\t]({1*0.47*cos(\t r)+0*0.47*sin(\t r)},{0*0.47*cos(\t r)+1*0.47*sin(\t r)});
\draw (1.04,1.92) node[anchor=north west] {$\scriptstyle \theta_{(\beta\alpha)}$};
\draw [->] (0.5,2) -- (-1.4,3.46);
\draw (-0.70,2.6) node[anchor=north west] {$\scriptstyle m_{\bar\beta}$};
\draw [shift={(0.5,2)}] plot[domain=1.13:2.49,variable=\t]({1*0.44*cos(\t r)+0*0.44*sin(\t r)},{0*0.44*cos(\t r)+1*0.44*sin(\t r)});
\draw (-0.3,3) node[anchor=north west] {$\scriptstyle \theta_{(\bar\beta\gamma)}$};
\draw [->] (-1.4,3.46) -- (1.18,3.46);
\draw (-0.35,3.9) node[anchor=north west] {$\scriptstyle m_\gamma$};
\end{tikzpicture}
\caption{The mass parallelogram for fusion process $(\a\b)\to\g$. The length of a vector is proportional to
the mass of the corresponding particle.}
 \label{fig:mass}
\end{figure}
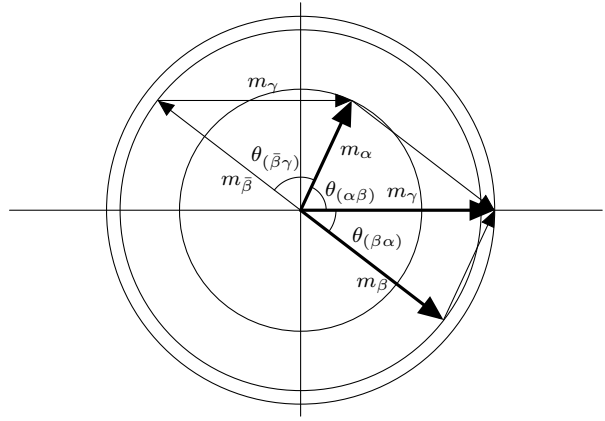

As a consequence of the assumptions above, there are certain relations between these quantities.

\begin{enumerate}
  \renewcommand{\theenumi}{(P\arabic{enumi})}
  \renewcommand{\labelenumi}{\theenumi}

\item {\bf Consequence of \ref{CPTinvariance}.}\label{p:CPT}
If $(\a\b)\to\g$ is a fusion, then there hold $\theta_{(\alpha \beta)} = \theta_{(\bar\alpha \bar\beta)}$, $\theta_{\alpha \beta} = \theta_{\bar\alpha \bar\beta}$ and $R_{\alpha \beta}^\g = R_{\bar\alpha \bar\beta}^{\bar\g}$.  

\item {\bf Consequence of Eq.~\eqref{prelations}.}\label{p:angles}
If $(\a\b)\to\g$ is a fusion, then there hold $\theta_{(\alpha \beta)} = \theta_{(\gamma \bar\beta)}$ and
$\theta_{\gamma \bar\beta}= \pi -\theta_{(\beta \alpha)}$ (see also Figure \ref{fig:mass}).

\item {\bf Consequence of \ref{crossing}.}
 \begin{equation}\label{RpR}
 R_{\alpha\beta}^{\prime\gamma}
 = \operatorname{Res}_{\zeta = i\pi -i\theta_{\alpha \beta}} S^{\bar\beta \alpha}_{\alpha \bar\beta}(\zeta)
 = \operatorname{Res}_{\zeta = i\pi -i\theta_{\alpha \beta}}S^{\alpha \beta}_{\beta \alpha}(i\pi -\zeta)
 = -R_{\alpha \beta}^\g. 
 \end{equation}

\item {\bf Consequence of \ref{bootstrap} and \ref{atzero}.}\label{p:bootstrap} Shifting $\zeta \rightarrow \zeta +i\theta_{(\beta \alpha)}$ in Eq.~\eqref{eq:bootstrap}, we find
 \begin{align*}
 S^{\gamma \nu}_{\nu \gamma}(\zeta +i\theta_{(\beta \alpha)}) = S^{\alpha \nu}_{\nu \alpha}(\zeta +i\theta_{\alpha \beta})S^{\beta \nu}_{\nu \beta}(\zeta).
 \end{align*}
 At $\nu = \beta$, we have
 \begin{align*}
 S^{\gamma \beta}_{\beta \gamma}(\zeta +i\theta_{(\beta \alpha)}) = S^{\alpha \beta}_{\beta \alpha}(\zeta +i\theta_{\alpha \beta})S^{\beta \beta}_{\beta \beta}(\zeta). 
\end{align*}
If $(\a\b)\to\g$ is a fusion, then so is $(\g\bar\b)\to\a$ and $S^{\g\b}_{\b\g}$ has a $t$-channel pole at $i\theta_{\g\bar\b}'$,
which is equal to $i\theta_{(\beta \alpha)}$ as we saw in \ref{p:angles}.
Furthermore, $S^{\alpha \beta}_{\beta \alpha}(\zeta +i\theta_{\alpha  \beta})$ has a pole at $\zeta =0$, while $S^{\b\b}_{\b\b}(\zeta)$
is unitary when $\zeta \in \RR$, therefore,
 \begin{align*}
 \operatorname{Res}_{\zeta= i\theta_{(\g \bar\b)}'} S^{\gamma \beta}_{\beta \gamma}(\zeta) = \operatorname{Res}_{\zeta =0} S^{\gamma \beta}_{\beta \gamma}(\zeta +i\theta_{(\beta \alpha)}) = \operatorname{Res}_{\zeta =i\theta_{\alpha \beta}} S^{\alpha \beta}_{\beta \alpha}(\zeta)S^{\beta \beta}_{\beta \beta}(0).
 \end{align*}
 Moreover, $S^{\beta \beta}_{\beta \beta}(0) = - 1$ by \ref{atzero}.
 Therefore, the equation above yields $R_{\gamma\bar\beta}^{\prime\alpha}= -R_{\alpha \beta}^\g$, which also implies
 by \eqref{RpR} that $R_{\gamma \bar\beta}^\a= R_{\alpha \beta}^\g$. 

 \item {\bf Consequence of simple fusion process.}\label{p:table}
We assumed that, if $(\a\b)\to\g$, then $\a\neq\g\neq\b$.
It follows that $R^\a_{\a\b} = R^\b_{\a\bar\a} = 0$.

\item {\bf Consequence of properties of elementary particles.}\label{p:elempart}
We assumed that any index $\alpha$ could be obtained with successive fusions of an elementary particle $\up$ (we use the notation $\alpha = \up^k$ for $k$ fusions, e.g. $(\up \up) \rightarrow \up^2$ ).  The following relations hold:
\begin{equation}\label{relations:elempart}
\theta_{(\up^k \up)}= \theta_0, \quad \theta_{(\up \up^k)} = k \theta_0, \quad m_{\up^k} =m_{\up} \frac{\sinh (k \theta_0)}{\sinh \theta_0},
\end{equation}
where $\theta_0$ is a constant (depending on the model). The first relation follows since $\theta_{(\alpha \up)}$ depends only on $\up$. The second and third relations are proved by induction on $k$. Specifically, the third relation is proved by taking the real part of Eq.~\eqref{prelations} with $\alpha = \up^k$, $\beta = \up$ and $\gamma = \up^{k+1}$ as below
\begin{equation*}
m_{\up^k} \cosh \zeta \cosh \theta_0 + m_{\up} \cosh \zeta \cosh \theta_{(\up \up^k)} = m_{\up^{k+1}} \cosh \zeta
\end{equation*}
and using the second relation in \eqref{relations:elempart}. The second relation follows by taking the imaginary part of Eq.~\eqref{prelations} with $\alpha = \up^{k+1}$, $\beta = \up$ and $\gamma= \up^{k+2}$, which is given by the following expression
 \begin{equation*}
m_{\up^{k+1}} \sinh \zeta \sinh \theta_0 - m_{\up} \sinh \zeta \sinh \theta_{(\up \up^{k+1})} =0
\end{equation*}
and using the third relation in \eqref{relations:elempart}.
\end{enumerate} 

We note that parity symmetry \ref{parity} implies that the bootstrap equation
holds for the flipped indices:
 \begin{equation*}
  S^{\nu\g}_{\g\nu}(\zeta) = S^{\nu\a}_{\a\nu}(\zeta +i\theta_{(\a\b)}) S^{\nu\b}_{\b\nu}(\zeta -i\theta_{(\b\a)}). 
 \end{equation*}

\subsection{Examples}\label{examples}

\subsubsection{The S-matrix of the \texorpdfstring{$Z(N)$}{Z(N)}-Ising model}\label{SzN}

The $Z(N)$-Ising models are conjectured integrable quantum field theories which should be obtained as the scaling limit of
certain statistical models \cite{KS79}. The form factors of fields in the models have been
investigated in \cite{BFK06}.

\begin{flushleft}
  \textbf{Particle spectrum}
\end{flushleft}
Let us fix $N \in \ZZ, N\ge 3$.
The $Z(N)$-Ising model contains $N-1$ species of particles, labeled by $\a \in \I = \{1,\cdots N-1\}$.
There are fusion processes where two particles of type $\alpha, \beta \in \I$
fuse into another particle of type $\alpha + \beta \mod N$, for all $\alpha + \beta \neq 0 \mod N$.
If $\alpha + \beta = N$, there is no fusion process. Moreover, the conjugate charge of a particle of type $\alpha$ is $\bar\alpha=N-\alpha$. 

The masses of the particles are 
\begin{equation}\label{masses}
m_\alpha = m_1 \frac{\sin \dfrac{\alpha\pi }{N}}{\sin \dfrac{\pi}{N}},
\end{equation}
where $m_1 > 0$ is the mass of the first particle, which is arbitrary.

\begin{flushleft}
  \textbf{Fusion table and angles}
\end{flushleft}
The angles $\theta_{(\alpha \beta)}$, yielding the position of the simple poles, can be obtained using \eqref{prelations} and \eqref{masses}.
One finds explicitly:
\begin{equation}\label{angles}
\theta_{(\alpha \beta)} = \begin{cases}  \frac{\beta\pi}{N},  &  \alpha + \beta < N, \\  \frac{(N-\beta)\pi}{N}, &  \alpha + \beta >  N.  \end{cases}
\end{equation}
Summarizing, the fusion table of the $Z(N)$-Ising model is (c.f\! \cite[Table 3.1]{Quella99},
there seems to be a mistake in the Bindungswinkel (fusion angle) of the second case):

\noindent
\begin{tabular}{|c|c|c|c|}
 \hline
 processes & rapidities of particles & fusion angles & cases\\
 \hline
 $(\a\b) \longrightarrow \a + \b $ & $\theta_{(\a\b)} = \frac{\b\pi}N, \theta_{(\b\a)} = \frac{\a\pi}N$ & $\theta_{\a\b} = \frac{(\a+\b)\pi}N$ & $\a + \b < N$ \\
 \hline 
 no fusion &  &  & $\a + \b = N$ \\
 \hline
 $(\a\b) \longrightarrow 2N - \a - \b$ & $\theta_{(\a\b)} = \frac{(N-\b)\pi}N, \theta_{(\b\a)} = \frac{(N-\a)\pi}N$ & $\theta_{\a\b} = \frac{(2N-\a-\b)\pi}N$ & $\a + \b > N$ \\
 \hline
\end{tabular}

\begin{flushleft}
  \textbf{S-matrix components and their poles}
\end{flushleft}
The component of the S-matrix corresponding to $\a = 1$ particles is given by 
\begin{align}\label{Sone:z}
S_{11}^{11}(\theta) = \dfrac{\sinh \dfrac{1}{2}\left(\theta +\dfrac{2\pi i}{N}\right)}{\sinh \dfrac{1}{2}\left(\theta -\dfrac{2i\pi}{N}\right)},
\end{align}
which has a simple pole at $\theta =i\theta_{11}= \frac{2 i\pi}{N}$ corresponding to the bound state $(11) \to 2$.
The value $\theta_{11}= \frac{2 i\pi}{N}$ can be computed using \eqref{prelations} and \eqref{masses} as well.

All the other S-matrix elements $S^{\beta \alpha}_{\alpha \beta}(\theta)$ can be computed using the Bootstrap equation \ref{bootstrap},
and one obtains the following expression (see \cite[Section 3.2]{Quella99} for a more detailed account of the models and
their bound states, yet one should be warned that the remark below about the product is missing there):
\begin{equation}\label{S}
S_{\beta \alpha}^{\alpha \beta}(\theta) = \prod_{m=-(\alpha -1)}^{ \alpha -1 *} \dfrac{\sinh \frac{1}{2}(\theta +\frac{i\pi}{N}(\beta+m+1)) \sinh \frac{1}{2}(\theta + \frac{i\pi}{N}(\beta+m -1))}{\sinh \frac{1}{2}(\theta -\frac{i\pi}{N}(\beta -m -1))\sinh \frac{1}{2}(\theta -\frac{i\pi}{N}(\beta -m +1))},
\end{equation}
where $*$ means that the index runs in steps of 2, namely $-(\a-1), -(\a-3),\cdots, \a-3, \a-1$.

If $\alpha + \beta < N$, this S-matrix element has two simple poles in the upper strip $\mathbb{R} +i(0,\pi)$
at $\zeta =i\theta_{\alpha \beta}= \frac{i\pi}{N}(\alpha + \beta)$, corresponding to the bound state $(\alpha \beta)\to \gamma = \a+\b$.
If furthermore $\a \neq \b$, there is another simple pole at $\zeta =i\theta'_{\beta\bar\alpha}= \frac{i\pi}{N}|\alpha -\beta|$.
There are double poles at a distance of $\frac{2i\pi}{N}$ between $\frac{i\pi}{N}|\alpha -\beta|$ and $\frac{i\pi}{N}(\alpha + \beta)$.

If $\alpha + \beta > N$, there is a simple pole at $\frac{i\pi}{N}(2N - \alpha - \beta)$,
and if $\a \neq \b$, another simple pole at $\frac{i\pi}{N}|\alpha -\beta|$,
and there are double poles at a distance of $\frac{2\pi i}{N}$ between them.

The properties \ref{unitarity}--\ref{polestructure} and \ref{regularity} are straightforward, and we check them in
\ref{withcdd} in more generality.
\ref{atzero} is also well-known to the experts, yet has not been very often used in computations.
To check it, one only has to note that in the product \eqref{S} with $\a=\b$ there is the cancellation
of a pole and a zero, and $S^{\a\a}_{\a\a}$ has only one simple pole in the physical strip (all the
other poles are double).

The properties of the elementary particles are new and we will check them below.

\begin{flushleft}\label{sec:elempartZ}
  \textbf{Elementary particles}
\end{flushleft}
We declare that the particles corresponding to the indices $1$ and $N-1$ are the elementary particles of the model.
In order to fix the notation, we take $\up = 1$.
It is clear that any index $\a \in I = \{1,\cdots, N-1\}$ can be achieved by $\a$ times fusion processes starting with $1$.

By \eqref{S}, the components 
\[
 S_{\beta 1}^{1 \beta}(\zeta) =
 \frac{\sinh \frac{1}{2}(\zeta +\frac{i\pi}{N}(\beta+1)) \sinh \frac{1}{2}(\zeta + \frac{i\pi}{N}(\beta-1))}{\sinh \frac{1}{2}(\zeta -\frac{i\pi}{N}(\beta-1))\sinh \frac{1}{2}(\zeta -\frac{i\pi}{N}(\beta+1))}
\]
have clearly only two simple poles at
$\zeta = \frac{i\pi}{N}(\beta-1)$ (if $\beta \neq 1)$ and $\zeta = \frac{i\pi}{N}(\beta+1)$ and no other poles.
Hence, by \ref{CPTinvariance}, also the components $S^{N-1\; \beta}_{\beta\; N-1}(\zeta) = S^{1 \bar\beta}_{\bar\beta 1}(\zeta)$
have only simple pole(s) at $\zeta = \frac{i\pi}{N}(N-\beta-1)$ (and at $\zeta = \frac{i\pi}{N}(N-\beta+1)$ if $\b \neq 1$).
The poles of the components $S^{\alpha \;N-1}_{N-1 \;\alpha}(\zeta)= S^{\bar \alpha 1}_{1 \bar \alpha}(\zeta)$ are also known from the above.
From these observation, it is clear that $S^{12}_{21}$ is the only component among $\{S^{1 \beta}_{\beta 1}\}$
which has a pole in $\RR + i[0,\frac{i\pi}N]$ (actually at $\frac{i\pi}N$).

Now, all the other properties of elementary particle are easy except
positivity of the residues. Actually, we can show that that $\operatorname{Res}_{\zeta =\frac{i\pi(\alpha + \beta)}{N}} S^{\alpha \beta}_{\beta \alpha}(\zeta) \in i \mathbb{R}_+$.

Let us first consider the case $\alpha \leq \beta$ and $\alpha + \beta <N$.
The pole at $\zeta = \frac{i\pi(\alpha + \beta)}{N}$ appears in Eq.~\eqref{S} in the sinh factor on the right hand side in the denominator for $m=-(\alpha -1)$. Hence, the residue at this pole is given by
\begin{multline*}
\operatorname{Res}_{\zeta =\frac{i\pi(\alpha + \beta)}{N}} S^{\alpha \beta}_{\beta \alpha}(\zeta) =\\
 2i \dfrac{\sin \frac{\pi}{N}(\beta+1) \sin \frac{\pi}{N}\beta}{\sin \frac{\pi}{N} }\prod_{m=-(\alpha -1)+2}^{ \alpha -1 *} \dfrac{\sin \frac{\pi}{2N}(\alpha +2\beta +m+1) \sin \frac{\pi}{2N}(\alpha +2\beta +m -1)}{\sin \frac{\pi}{2N}(\alpha + m +1)\sin \frac{\pi}{2N}(\alpha +m -1)},
\end{multline*}
Since $-(\alpha -1)+2\leq m \leq \alpha -1$ and $1\leq \alpha, \beta <N-1, \a + \b < N$ (the last one by assumption),
then all the arguments of the sine in the above formula are between $0$ and $\pi$, where the sine is indeed positive.

The case $\alpha > \beta$ follows from parity \ref{parity} and
the case $\alpha + \beta >N$ is automatic using CPT \ref{CPTinvariance}.

\subsubsection{S-matrix with CDD factor}\label{withcdd}

We treat more examples of S-matrices with the same mass spectrum and the fusion table as those of the $Z(N)$-Ising model.
In this class of diagonal $S$-matrices, given an $S^{11}_{11}$ with suitable properties,
one can construct the full $S$-matrix by formula \eqref{generalS} below
and show that it fulfills properties of Section \ref{S-matrix}.

\paragraph{General properties}
As the model has the same mass spectrum and fusion processes as the $Z(N)$-Ising model, $S$ must have the same pole structure as
the S-matrix ${S_{Z(N)}}$ of that model. Hence, we can write
$S^{11}_{11}(\zeta) = {S_{Z(N)}}^{11}_{11}(\zeta){S_{\textrm{CDD}}}^{11}_{11}(\zeta)$, where ${S_{\textrm{CDD}}}^{11}_{11}(\zeta)$ has no pole in the physical strip. 
${S_{\textrm{CDD}}}^{11}_{11}$ is called a CDD factor \cite{CDD56}.
We assume that $S^{11}_{11}$ satisfies
\begin{itemize}
 \item Unitarity \ref{unitarity},
 \item Hermitian analyticity \ref{hermitian},
 \item Periodicity $S_{11}^{11}(\zeta) = S_{11}^{11}(\zeta + 2\pi i)$,
 \item Bootstrap consistency $\prod_{j=0}^{N-1}S^{11}_{11}\left( \zeta + \frac{2\pi i}{N}j \right) = 1$.
 \item $S^{11}_{11}(0) = -1$,
 \item $\lim_{\epsilon \searrow 0} \; S^{11}_{11}\left(\epsilon + \frac{2k i \pi}{N}\right) S^{11}_{11}\left( \epsilon - \frac{2 k \pi i}{N}\right) = \begin{cases}  -1,  & \text{ if }  k=1, \\  +1, & \text{ if }  k>1,  \end{cases}$
 \item  $S^{11}_{11}$ has only finitely many zeros in the physical strip and they are configured as follows:
\begin{enumerate}
         \item Fix an integer $1\le k \le N-1$. For $B \in \CC$ such that $k-1 < \re B < k$, the multiplicities of zeros of
         $S_{11}^{11}$ at $\frac{i\pi B}N, \frac{i\pi \overline{B}}N, \frac{i\pi (2k-B)}N, \frac{i\pi (2k-\overline{B})}N$
         are the same (the multiplicities can be $0$).
         \item Fix an integer $1\le k \le N-1$. For $B \in \CC$ such that $\re B = k$, the multiplicities of zeros of
         $S_{11}^{11}$ at $\frac{i\pi B}N, \frac{i\pi \overline{B}}N$
         are the same.
         \item Fix an odd integer $1 \le k \le N-1$. For $B \in \RR$ such that $k-1 < \re B \le k$, the multiplicities of zeros of
         $S_{11}^{11}$ at $\frac{i\pi B}N, \frac{i\pi (2k-B)}N$
         are the same. If $\re B = k$, it should have a zero of even degree.
         \item $S^{11}_{11}$ does not have zeros at $\frac{2\pi k i}N$, $k\in\ZZ$.
       \end{enumerate}
 \item $ {S_{\textrm{CDD}}}^{11}_{11}( \zeta ) \in \RR_+$ at the points $\zeta = \frac{2\pi pi}{N}, p \in \ZZ$.
\item $S^{11}_{11}$ is bounded both below and above in a neighborhood of the point of infinity.
\end{itemize}
There are meromorphic functions which have zeros exactly at those points specified above.
They are called Blaschke products and can be explicitly written, correspondingly to the three cases above, as
\begin{align*}
 {S_B}_{11}^{11}(\zeta)_1 &:= \frac{\sinh \frac{1}{2}\left( \zeta -\frac{i\pi B}{N}  \right)\sinh \frac{1}{2}\left( \zeta -\frac{i\pi \overline{B}}{N}  \right)\sinh \frac{1}{2} \left( \zeta - \frac{i\pi(2k-B)}{N} \right)\sinh \frac{1}{2} \left( \zeta - \frac{i\pi(2k-\overline{B})}{N}\right)}{\sinh \frac{1}{2}\left( \zeta +\frac{i\pi B}{N}  \right)\sinh \frac{1}{2}\left( \zeta +\frac{i\pi \overline{B}}{N}  \right)\sinh \frac{1}{2} \left( \zeta + \frac{i\pi(2k-B)}{N} \right)\sinh \frac{1}{2} \left( \zeta + \frac{i\pi(2k-\overline{B})}{N} \right)} \\
 {S_B}_{11}^{11}(\zeta)_2 &:= \frac{\sinh \frac{1}{2}\left( \zeta -\frac{i\pi B}{N}  \right)\sinh \frac{1}{2}\left( \zeta -\frac{i\pi \overline{B}}{N}\right)}{\sinh \frac{1}{2}\left( \zeta +\frac{i\pi B}{N}  \right)\sinh \frac{1}{2}\left( \zeta +\frac{i\pi \overline{B}}{N}\right)} \\
 {S_B}_{11}^{11}(\zeta)_3 &:= \frac{\sinh \frac{1}{2}\left( \zeta -\frac{i\pi B}{N}  \right)\sinh \frac{1}{2}\left( \zeta -\frac{i\pi(2k-B)}{N} \right)}{\sinh \frac{1}{2}\left( \zeta +\frac{i\pi B}{N}  \right)\sinh \frac{1}{2} \left( \zeta + \frac{i\pi(2k-B)}{N} \right)}
\end{align*}
We will later on refer to these Blaschke products as simply ${S_B}_{11}^{11}(\zeta)$.

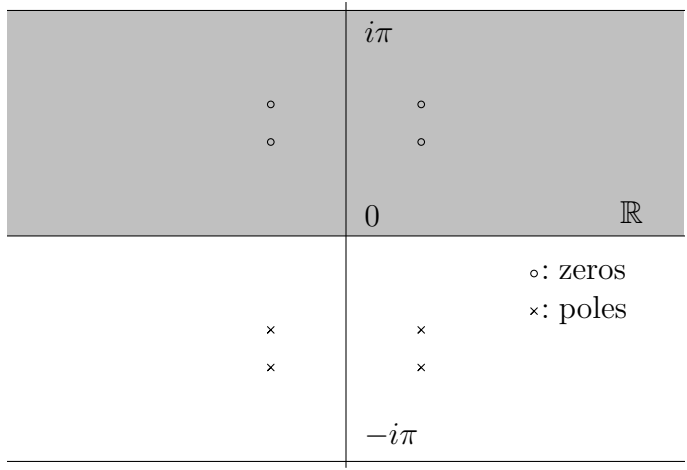
\begin{figure}[ht]
\centering
\definecolor{qqqqff}{rgb}{0,0,0}
\begin{tikzpicture}[line cap=round,line join=round,>=triangle 45,x=1.0cm,y=1.0cm, scale=0.5]
\clip(-9,-6.2) rectangle (9,6.2);
\fill[line width=0pt,fill=black,fill opacity=0.25] (-9,6) -- (-9,0) -- (9,0) -- (9,6) -- cycle;
\draw [domain=-9.5:11.04] plot(\x,{(-0-0*\x)/2});
\draw (0,-6.16) -- (0,6.22);
\draw [domain=-9.5:11.04] plot(\x,{(--48-0*\x)/8});
\draw [domain=-9.5:11.04] plot(\x,{(-48-0*\x)/8});
\draw (7,1.22) node[anchor=north west] {$\mathbb{R}$};
\draw (0.2,6) node[anchor=north west] {$i\pi$};
\draw (0.2,1.1) node[anchor=north west] {$0$};
\draw (0.2,-4.7) node[anchor=north west] {$-i\pi$};
\draw (4.9,-0.5) node[anchor=north west] {$: \textrm{zeros}$};
\draw (4.9,-1.3) node[anchor=north west] {$: \textrm{poles}$};
\begin{scriptsize}
\draw [color=qqqqff] (-2,2.5) circle (2.5pt);
\draw [color=qqqqff] (2,2.5) circle (2.5pt);
\draw [color=qqqqff] (2,3.5) circle (2.5pt);
\draw [color=qqqqff] (-2,3.5) circle (2.5pt);
\draw [color=qqqqff] (2,-2.5)-- ++(-2.5pt,-2.5pt) -- ++(5.0pt,5.0pt) ++(-5.0pt,0) -- ++(5.0pt,-5.0pt);
\draw [color=qqqqff] (2,-3.5)-- ++(-2.5pt,-2.5pt) -- ++(5.0pt,5.0pt) ++(-5.0pt,0) -- ++(5.0pt,-5.0pt);
\draw [color=qqqqff] (-2,-2.5)-- ++(-2.5pt,-2.5pt) -- ++(5.0pt,5.0pt) ++(-5.0pt,0) -- ++(5.0pt,-5.0pt);
\draw [color=qqqqff] (-2,-3.5)-- ++(-2.5pt,-2.5pt) -- ++(5.0pt,5.0pt) ++(-5.0pt,0) -- ++(5.0pt,-5.0pt);
\draw [color=qqqqff] (5,-1) circle (2.5pt);
\draw [color=qqqqff] (5,-2)-- ++(-2.5pt,-2.5pt) -- ++(5.0pt,5.0pt) ++(-5.0pt,0) -- ++(5.0pt,-5.0pt);
\end{scriptsize}
\end{tikzpicture}
\caption{Zeros and poles of the Blaschke factor ${S_B}_{11}^{11}$ with $2 < \re B < 3, \im B \neq 0$ and $N = 6$.
The shaded area is the physical strip.}
 \label{fig:blaschke}
\end{figure}

Bootstrap consistency is an equality between meromorphic functions, hence we can translate the variable $\zeta$
and obtain
\begin{equation}\label{identity}
\prod_{j=l+1}^{2N+l-1 *}S^{11}_{11}\left( \zeta + \frac{i\pi}{N}j \right) = 1
\end{equation}
for arbitrary $l$, by rewriting the product in steps of two. Note that this identity implies that all indices can be read \emph{mod N},  i.e.
\begin{multline*}
S^{k \ell}_{\ell k}(\zeta) =  \prod_{m = -(k -1)}^{k -1 *} \prod_{n = -(\ell-1)}^{\ell-1 *} S^{11}_{11}\Big( \zeta  + \frac{i\pi (m+n)}{N} \Big) \\
=  \prod_{m = -(k -1)}^{k -1 *} \prod_{n = -(N+\ell-1)}^{N+\ell-1 *} S^{11}_{11}\Big( \zeta  + \frac{i\pi (m+n)}{N} \Big) = S^{k,  \ell +N}_{\ell +N, k}(\zeta).
\end{multline*}

\paragraph{Examples.}
Functions $S^{11}_{11}$ of the following form satisfy the properties specified above:
\begin{equation}\label{eq:CDD}
S^{11}_{11}(\zeta) = {S_{Z(N)}}^{11}_{11}(\zeta) \prod_j {S_{B_j}}_{11}^{11}(\zeta),
\end{equation}
where ${S_{Z(N)}}^{11}_{11}$ is the matrix component of the $Z(N)$-Ising model we studied in Section \ref{SzN},
${S_{B_j}}^{11}_{11}$ are the Blaschke products explained above and the product is finite, so that $S^{11}_{11}$
has only finitely many zeros and is bounded in a neighborhood of the point of infinity.
For this particular example, the CDD factor is the product part:
${S_{\textrm{CDD}}}^{11}_{11}( \zeta ) = \prod_j {S_{B_j}}_{11}^{11}(\zeta)$.

It is straightforward to check the above mentioned properties of this $S^{11}_{11}$.
Indeed, the first five are easy by noting that ${S_{B}}^{11}_{11}(0) = 1$.
The limit
\[
 \lim_{\epsilon \searrow 0} \; S^{11}_{11}\left(\epsilon + \frac{2k i \pi}{N}\right) S^{11}_{11}\left( \epsilon - \frac{2 k \pi i}{N}\right)
\]
can be also separately checked: for ${S_{Z(N)}}^{11}_{11}$ it is straightforward to check that the limit is $-1$
and for ${S_{B}}^{11}_{11}$ the function is continuous and the limit is $1$ by unitarity and hermitian analyticity.
The properties of zeros are exactly encoded in the Blaschke products.
For the last property, we note that ${S_{B}}^{11}_{11}(t) \in \RR$ for $t\in i\RR$ by hermitian analyticity
and it is continuous. Furthermore, as we saw, ${S_{B}}^{11}_{11}(0) = 1$.
Finally, in order to check the sign at $\frac{2\pi pi}{N}$, $p\in \ZZ$, we only have to
note that ${S_{B}}^{11}_{11}(t)$ has always even number of poles (including multiplicity)
between each interval $(\frac{i\pi(k-1)}N,\frac{i\pi(k+1)}N)$, where $k$ is odd. Then the value at
$\frac{2\pi pi}{N}$ must be positive.

\paragraph{Full S-matrix}

For such a given $S_{11}^{11}$, we construct the full S-matrix components as
\begin{equation}\label{generalS}
S^{k \ell}_{\ell k}(\zeta) = \prod_{m = -(\ell -1)}^{\ell -1 *} \prod_{n = -(k-1)}^{k-1 *} S^{11}_{11}\left( \zeta + \frac{i\pi (m+n)}{N} \right),
\end{equation}
where $*$ means that the index runs in steps of $2$. It is now obvious that the $S$-matrix is uniquely fixed by $S^{11}_{11}$.

Let us now proceed to verify properties \ref{unitarity}--\ref{atzero} for the full S-matrix components
from these properties.

\paragraph{\ref{unitarity}}  We compute
\begin{equation*}
\overline{S^{\ell k}_{k \ell}(-\bar\zeta)} 
=\prod_{m = -(k -1)}^{k -1 *} \prod_{n = -(\ell-1)}^{\ell-1 *} \overline{S^{11}_{11}\left( -\left( \overline{\zeta  + \frac{i\pi (m+n)}{N}}\right) \right)}= S^{k \ell}_{\ell k}(\zeta),
\end{equation*}
where the first equality follows from renaming the indices $m$ into $-m$ and $n$ into $-n$ and by reordering the product, and in the last equality we assumed that $S^{11}_{11}$ fulfills \ref{unitarity}.

\paragraph{\ref{parity}} This follows from the fact that the indices $k,\ell,m,n$ appear symmetrically in the following expression
\begin{align*}
S^{\ell k}_{k \ell}(\zeta) =  \prod_{m = -(k -1)}^{k -1 *} \prod_{n = -(\ell-1)}^{\ell-1 *} S^{11}_{11}\left( \zeta  + \frac{i\pi (m+n)}{N} \right) = S^{k \ell}_{\ell k}(\zeta).
\end{align*}

\paragraph{\ref{hermitian}} This follows from the following computation:
\begin{equation*}
S^{\ell k}_{k \ell}(-\zeta)^{-1} 
                                                        =  \prod_{m = -(k -1)}^{k -1 *} \prod_{n = -(\ell-1)}^{\ell-1 *} S^{11}_{11}\left( \zeta - \frac{i\pi (m+n)}{N} \right)
= S^{k \ell}_{\ell k}(\zeta),
\end{equation*}
where in the first equality we used that $S^{11}_{11}$ fulfills property \ref{hermitian}
and in the last equality we renamed the indices $m$ into $-m$ and $n$ into $-n$.

\paragraph{\ref{crossing}} We compute
\begin{align*}
S^{\bar k \ell}_{\ell \bar k}(\zeta) 
                                                                 &= \prod_{m = -(\ell -1)}^{\ell -1 *} \prod_{n = N-k+1}^{N+k-1 *} S^{11}_{11}\left( \zeta  + \frac{i\pi (m+n)}{N} \right)^{-1}\\
                                                                 &= \prod_{m = -(\ell -1)}^{\ell-1 *} \prod_{n' = -(k-1)}^{k-1 *} S^{11}_{11}\left( -i\pi -\zeta + \frac{i\pi (m+n')}{N} \right)  
= S^{\ell k}_{k \ell}(i\pi -\zeta),
\end{align*}
where in the first equality we used \eqref{identity},
in the second equality we renamed $n-N$ into $n'$
and used the fact that $S^{11}_{11}$  fulfills property \ref{hermitian},
and the last equality follows from the $2\pi i$-periodicity of $S^{11}_{11}$.

\paragraph{\ref{CPTinvariance}} We compute
\begin{align*}
S^{\bar k \bar \ell}_{\bar \ell \bar k}(\zeta) 
                                                                                   &=\prod_{m = -(N-\ell -1)}^{N-\ell -1 *} \prod_{n =N-k+1}^{N+k-1 *} S^{11}_{11}\left( \zeta + \frac{i\pi (m+n)}{N} \right)^{-1}\\
                                                                                   &=\prod_{m = N-\ell +1}^{N+\ell -1 *} \prod_{n = N-k+1}^{N+k-1 *} S^{11}_{11}\left( \zeta + \frac{i\pi (m+n)}{N} \right)\\
                                                                                   &=\prod_{m' = -(\ell -1)}^{\ell -1 *} \prod_{n' = -(k-1)}^{k-1 *} S^{11}_{11}\left( \zeta + \frac{i\pi (n'+m')}{N} \right)= S^{k \ell}_{\ell k}(\zeta),  
\end{align*}
where in the first and second equalities we used \eqref{identity},
in the third equality we renamed the indices $n-N$ into $n'$ and $m-N$ into $m'$, 
and used the $2\pi i$-periodicity of $S^{11}_{11}$.

\paragraph{\ref{bootstrap}}
In order to check the bootstrap equation,
we compute the right-hand side of the first Equation in \eqref{eq:bootstrap} in the case where $\alpha + \beta < N$:
\begin{align*}
&S^{\alpha \nu}_{\nu \alpha}(\zeta +i\theta_{(\alpha \beta)}) S^{\beta \nu}_{\nu \beta}(\zeta -i\theta_{(\beta \alpha)}) \\
&\quad =\prod_{m = -(\nu -1)}^{\nu -1 *}   \left\lbrack  \prod_{n = -(\alpha -1)}^{\alpha -1 *} S^{11}_{11}\left( \zeta + \frac{i\pi (\beta + m +n)}{N} \right)\prod_{n' = -(\beta-1)}^{\beta-1 *} S^{11}_{11}\left( \zeta + \frac{i\pi (m+n' -\alpha)}{N} \right)\right\rbrack \\
&\quad =\prod_{m = -(\nu -1)}^{\nu -1 *}   \left\lbrack  \prod_{\tilde n = \beta -\alpha +1}^{\beta +\alpha -1 *} S^{11}_{11}\left( \zeta + \frac{i\pi (m + \tilde n)}{N} \right)\prod_{\tilde n'= -(\beta + \alpha -1)}^{\beta -\alpha-1 *} S^{11}_{11}\left( \zeta + \frac{i\pi (m+\tilde n')}{N} \right)\right\rbrack
= S^{\gamma \nu}_{\nu \gamma}(\zeta),
\end{align*} 
where in the second equality we used the renaming $\tilde n = n + \beta$ and $\tilde n' = n' -\alpha$ and in the last equality we combined the two products in the square brackets and called $\gamma = \alpha + \beta$. In the case where $\alpha + \beta > N$, we can use the CPT property on both sides of  Eq.~\eqref{eq:bootstrap} to reduce this case to the previous case with $\bar \alpha + \bar \beta < N$.

\paragraph{\ref{polestructure}}
We have to prove that $S^{kl}_{lk}$ does not have any pole in the physical strip other than those of ${S_{Z(N)}}^{kl}_{lk}$.

Now, using Eq.~\eqref{generalS}, we deduce the poles and zeros for the more general matrix element $S^{k \ell}_{\ell k}$
starting from $S^{11}_{11}$. To this end, we start by considering the matrix element $S^{1 \ell}_{\ell 1}$ with $\ell >1$,
\begin{equation}\label{Sell}
S^{1 \ell}_{\ell 1}(\zeta) = \prod_{n = -(\ell -1)}^{\ell -1 *} S^{11}_{11}\left( \zeta + \frac{i\pi n}{N} \right).
\end{equation}
Assume that $S_{11}^{11}$ has a zero at $\frac{i\pi B}N$,
hence $S_{11}^{11}/{S_B}_{11}^{11}$ is still analytic in the physical strip.
Now, because of unitarity and hermitian analyticity of $S_{11}^{11}$, it must have poles outside the
physical strip corresponding to the zeros of ${S_B}_{11}^{11}$ and they
might enter in the physical strip in \eqref{Sell}. What we have to do is to exclude this possibility.

The factor responsible for this concern is
\begin{equation}\label{SBell}
{S_B}^{1 \ell}_{\ell 1}(\zeta) = \prod_{n = -(\ell -1)}^{\ell -1 *} {S_B}^{11}_{11}\left( \zeta + \frac{i\pi n}{N} \right).
\end{equation}
Assume that $k-1 < \re B < k$. Then the factors ${S_B}^{11}_{11}\left( \zeta + \frac{i\pi n}{N} \right)$ have
poles in the physical strip for $-(\ell-1)\le n\le -k$ or $\ell-1 \ge n \ge N-k$ (the latter comes from periodicity).
For such $n$, there is exactly $m:=2k+n\le \ell-1$ or $-(\ell-1)\le 2(k-N)-n=:m$ (respectively to the cases above,
and note that $2k$ is even hence these $m$ are in the steps of the product)
such that the poles of ${S_B}^{11}_{11}\left( \zeta + \frac{i\pi n}{N}\right)$
are canceled by zeros of ${S_B}^{11}_{11}\left( \zeta + \frac{i\pi m}{N}\right)$.

The cases with $\re B = 0$ and $\im B = 0$ (namely, the cases 1. and 2. in Sec.~\ref{withcdd}, respectively) can be argued analogously.
Therefore, each Blaschke product ${S_B}_{11}^{11}$ contained in $S_{11}^{11}$ does not bring
new poles into $S^{1 \ell}_{\ell 1}$ in the physical strip and
$S^{1\ell}_{\ell 1}(\zeta)$ has single poles at $\frac{(\ell +1)i\pi }{N},  \frac{(\ell -1)i\pi}{N}$

Finally, for more general components, we have
\begin{equation*}
S^{k \ell}_{\ell k}(\zeta) = \prod_{m = -(k -1)}^{k -1 *} S^{1 \ell}_{\ell 1}\left( \zeta + \frac{i\pi m}{N} \right)
\end{equation*}
and we can again extract the zeros of $S^{1 \ell}_{\ell 1}$, which are of the form of a Blaschke product,
and repeat the same argument to conclude that $S^{k \ell}_{\ell k}$ does not have any extra pole in the physical strip.

\paragraph{\ref{atzero}} We consider the matrix element
\begin{equation}\label{zero1}
 S^{k k}_{k k}(\zeta) = \prod_{m = -(k -1)}^{k -1 *}    \prod_{n = -(k -1)}^{k -1 *} S^{11}_{11}\left( \zeta + \frac{i\pi ( m +n)}{N} \right).
\end{equation}
By renaming of the indices $m,n$ into $m+n = a$ and $m-n =b$,
so that $ -2(k-1)\leq a \leq 2(k-1)$ and $-\frac{2(k-1) -\lvert a \rvert}{2}\leq b \leq \frac{2(k-1) - \lvert a \rvert}{2}$
(note that $a$ is always even irrespective of $k$), we find 
\begin{equation}\label{zero3}
\eqref{zero1} =  \prod_{a = -2(k -1)}^{2(k -1) *}    \prod_{b = -(k -1) + \frac{\lvert a \rvert}{2}}^{(k -1) - \frac{\lvert a \rvert}{2} *} S^{11}_{11}\left( \zeta + \frac{i\pi a}{N} \right)
=  \prod_{a = -2(k -1)}^{2(k -1) *}   S^{11}_{11}\left( \zeta + \frac{i\pi a}{N} \right)^{c(a)},
\end{equation}
where $c(a) = k-1 -\frac{\lvert a \rvert}{2}$ or $k- \frac{\lvert a \rvert}{2}$ depending on
whether $k-1 -\frac{\lvert a \rvert}{2}$ is odd or even.
We evaluate this product at $\zeta = \epsilon$ with $\epsilon \searrow 0$. Using the assumptions
on $S_{11}^{11}$, we find from \eqref{zero3},
\begin{align*}
\lim_{\epsilon \searrow 0} S^{kk}_{kk}(\epsilon)
&= \lim_{\epsilon \searrow 0 } \; (S(0)^{11}_{11})^{c(0)}
\prod_{l=1}^{k-1} \left\lbrack  S^{11}_{11}\left( \epsilon -\frac{2\pi li }{N} \right) S^{11}_{11}\left( \epsilon + \frac{2\pi li}{N} \right)  \right\rbrack^{c(2l)-1} \\
&=(-1)^{c(0)}(-1)^{c(2)} = (-1)^{2c(0)-1} = -1.
\end{align*}

\paragraph{\ref{regularity} }  Eq.~\eqref{generalS} represents $S^{kl}_{lk}$ as the finite product of
$S^{11}_{11}$, therefore, it has only finitely many zeros in the physical strip. We checked \ref{unitarity}, hence $S^{kl}_{lk}$ is bounded in a neighborhood of an arbitrary finite interval in $\RR$. Again by Eq.~\eqref{generalS} and
the assumption that $S^{11}_{11}$ is bounded in a neighborhood of the point of infinity, $S^{kl}_{lk}$
is bounded in a neighborhood of the whole real line.

\paragraph{Elementary particle.}
As in $Z(N)$-Ising model, the elementary particle is $\up = 1$.
All the properties are easy or have been checked except positivity of the residues.

Let us consider the case with $k \leq \ell$ and $k+\ell <N$. Then $S^{k \ell}_{\ell k}(\zeta)$ has a simple pole at $\zeta = \frac{(k + \ell)\pi i}{N}$ with residue given by
\begin{equation*}
\operatorname{Res}_{\zeta = \frac{(k + \ell)\pi i}{N}}S^{k \ell}_{\ell k}(\zeta) =
\operatorname{Res}_{\zeta = \frac{(k + \ell)\pi i}{N}}{S_{Z(N)}}^{k \ell}_{\ell k}(\zeta)
\cdot  {S_{\textrm{CDD}}}^{k \ell}_{\ell k}\left( \frac{i(k + \ell)\pi }{N}\right),
\end{equation*}
where  ${S_{\textrm{CDD}}}^{k \ell}_{\ell k}(\zeta) :=S^{k \ell}_{\ell k}(\zeta) /{S_{Z(N)}}^{k \ell}_{\ell k}(\zeta)$
is analytic at the poles of ${S_{Z(N)}}^{k \ell}_{\ell k}(\zeta)$. 
(By crossing symmetry there is another simple pole at $\zeta = \frac{(\ell -k)\pi i}{N}$ with residue with opposite sign.)

By \eqref{generalS} we can write
\begin{equation}\label{Sredprod}
{S_{\textrm{CDD}}}^{k \ell}_{\ell k}(\zeta) = \prod_{m = -(\ell -1)}^{\ell -1 *} \prod_{n = -(k -1)}^{k -1 *}
{S_{\textrm{CDD}}}^{11}_{11}\left( \zeta + \frac{i\pi (m +n)}{N} \right),
\end{equation}
where ${S_{\textrm{CDD}}}^{11}_{11}(\zeta) = S^{11}_{11}(\zeta) / {S_{Z(N)}}^{11}_{11}(\zeta)$.
In Sec.~\ref{sec:elempartZ} we showed that
$ \operatorname{Res}_{\zeta = \frac{(k + \ell)\pi i}{N}}{S_{Z(N)}}^{k \ell}_{\ell k}(\zeta)  \in i\mathbb{R}_+$.
Hence, it only remains to prove that under certain assumptions on $S^{11}_{11}$,
the product \eqref{Sredprod} evaluated at $\zeta = \frac{(k + \ell)\pi i}{N}$ is real and non-negative.
We note that at this value of $\zeta$ the argument of $S^{11}_{11}$ in Eq.~\eqref{Sredprod} is
$i \lambda := \frac{i\pi (k + \ell + m +n)}{N}$ with $\frac{2\pi}{N}\leq  \lambda\leq \frac{2\pi (k + \ell -1)}{N}$,
since $-(\ell -1)\leq m\leq \ell -1$ and $-(k-1)\leq n\leq k-1$. We can write $i\lambda = \frac{2 i \pi p}{N}$ with
$p \in \mathbb{N}$,  $1\leq p \leq N-1$. Indeed, if $k$ is even/odd then $n$ is odd/even respectively,
and if $\ell$ is even/odd then $m$ is odd/even respectively, so that the sum $k + \ell + m+n$ is always an even integer.
By assumption, ${S_{\textrm{CDD}}}^{k \ell}_{\ell k}(\zeta)$ is also real and positive.
The case $k >\ell$ follows by \ref{parity} and the case $k + \ell >N$ by using \ref{CPTinvariance}.

\subsubsection{Affine Toda field theory}\label{toda}
Certain concrete models with Lagrangians, called $A_{N-1}$-affine Toda field theories \cite{BCKKS92},
are believed to be associated with some of the examples
of functions $S^{11}_{11}$ which fulfill the properties of Section \ref{withcdd}:
These theories are conjectured to possess the same particle content
and fusion processes of the $Z(N)$-Ising model and the scattering of two particles of type $1$ is given by
\begin{equation}\label{Stoda}
{S_{A_{N-1}}}^{11}_{11}(\zeta) = {S_{Z(N)}}^{11}_{11}(\zeta) \frac{\sinh \frac{1}{2}\left( \zeta -\frac{i\pi}{N}B  \right)\sinh \frac{1}{2} \left( \zeta - \frac{i\pi}{N}(2-B) \right)}{\sinh \frac{1}{2}\left( \zeta + \frac{i\pi}{N}B \right)\sinh \frac{1}{2}\left( \zeta + \frac{i\pi}{N}(2-B) \right)}, 
\end{equation}
where $ 0\leq B \leq 1$ is a number related to the coupling constant in the Lagrangian \cite{BK03, Korff:2000},
and ${S_{Z(N)}}^{11}_{11}$ is the matrix component of the $Z(N)$-Ising model we studied in Section \ref{SzN}.
It is just one example of functions pointed out in Section \ref{withcdd}
with a single CDD factor of the third case where $B$ is real.

Affine Toda field theories should have quantum-group symmetry and it is also conjectured
that their scaling limit gives rise to a well-known family of conformal field theory \cite{Korff:2000}.
Such a connection would be interesting in the operator-algebraic approach, c.f.\! \cite{BT13}.
We hope to come back to the whole family of affine Toda field theories and study these aspects
in future publications.

\subsubsection{General comments}

The properties in Section \ref{S-matrix} are not completely general.
The tensor product of models is obviously possible, but we excluded this
by assuming that there is only one pair of elementary particles.
One can also twist a tensor product by an analytic function as \cite{Tanimoto14-1}.
Yet, we do not know whether there is any (candidate for) diagonal S-matrix which has more than one pair
of elementary particles and does not fall in these classes.

Furthermore, the requirement that a pair of particles fuses into only one species
is also special. In the sine-Gordon model, more complicated fusion processes occur,
yet breather-breather S-matrices are diagonal. We will come back to this point
in a future publication \cite{CT-sine}.

\subsection{Single-particle space and S-symmetric Fock space}\label{Fock}

Following partially \cite[Section 2]{LS14}, we generalize the Hilbert space construction of \cite{Lechner08} to these multi-particle S-matrix models.
As we introduced in Section \ref{S-matrix}, the particle species (charged or uncharged) are labeled by an index $\alpha \in \I$ and have masses $m_{\alpha}>0$.
Therefore, the single particle Hilbert space $\mathcal{H}_1$ is the direct sum of all species $\alpha$:
\begin{align*}
\mathcal{H}_1 = \displaystyle{\bigoplus_{\alpha \in \I}} \mathcal{H}_{1,\alpha}, \quad \mathcal{H}_{1,\alpha}= L^2(\mathbb{R},d\theta).
\end{align*}
An element $\Psi_1 \in \mathcal{H}_1$ can be identified as a $K$-component vector valued function $\theta \mapsto \Psi_1^{\alpha}(\theta)$,
where $K = |\I|$.

On $\mathcal{H}_1$, there is a unitary representation of the proper orthochronous Poincar\'e group $\poincare$:
\begin{align*}
  U_1 := \bigoplus_{\alpha \in \I} U_{1,m_{\alpha}},
  \quad (U_{1,m_\a}(a,\l)\Psi^\a_1)(\theta) := \exp \left(i p_{m_\a}(\theta) \cdot a\right) \Psi^{\a}_1(\theta-\l),
\end{align*}
where the momentum of the particle $p_{m_{\alpha}}(\theta)$ is defined as in Section \ref{S-matrix}.

As already explained, the particles may carry a charge, and we denote the conjugate charge of $\alpha$ by $\bar\alpha$.
The corresponding CPT operator acts on $\mathcal{H}_{1,\alpha}$ by the antiunitary representation $U_1(j):= J_1$
(where $j:\; \pmb{x}\rightarrow -\pmb{x}$ is an element of the proper Poincar\'e group),
\begin{equation}\label{J1}
(J_1 \Psi_1)^{\alpha}(\theta) := \overline{\Psi^{\bar\alpha}_1(\theta)}.
\end{equation}

Now that the S-matrix and its properties have been introduced, the Hilbert space of the theory is constructed in analogy with \cite{LS14}.
We have by now considered $\{S^{\a\b}_{\g\d}\}$ for $\a = \d$ and $\b = \g$. We extend it to a matrix-valued function
by $S^{\alpha \beta}_{\gamma \delta}(\theta) = S^{\alpha \beta}_{\beta \alpha}(\theta) \delta_{\delta}^{\alpha} \delta_{\gamma}^{\beta}$.
For (almost every) fixed $\theta$ , this is a $K^2\times K^2$-matrix and can be identified with an operator on $\CC^{K^2}$.
By this construction, the S-matrix $S$ is diagonal in the following sense (c.f.\! \cite{BT15, Tanimoto14-1}). Introducing 
the R-matrix by $R^{\a\b}_{\g\d}(\theta) := S^{\b\a}_{\g\d}(\theta)$, it is
a diagonal matrix: $R_{\g\d}^{\a\b}(\theta) = 0$ unless $\a = \g$ and $\b = \d$.

For $n \in \NN$, we consider the tensor products $\mathcal{H}_1^{\otimes n}$, and define a representation $D_n$ of the symmetric group $\mathfrak{S}_n$
on it, acting as
\begin{equation*}
  (D_n(\tau_k)\Psi_n)^{\pmb{\a}}(\pmb{\theta}) = S_{\a_{k+1}\a_k}^{\a_k\a_{k+1}}(\theta_{k+1}-\theta_k)\Psi^{\a_1\cdots\a_{k+1}\a_k\cdots\a_n}_n(\theta_1,\cdots, \theta_{k+1},\theta_k,\cdots,\theta_n),
\end{equation*}
$\pmb{\theta}:=(\theta_1,\ldots,\theta_n), \pmb{\a} := (\a_1\cdots \a_n)$ and
$\tau_k \in \mathfrak{S}_n$ is the transposition $(k,k+1) \mapsto (k+1,k)$.
Since our S-matrix is diagonal, hence the Yang-Baxter equation is trivial, and since it satisfies the requirements of \cite[Definition 2.1(i, ii, iii)]{LS14}, then $D_n$ extends to a unitary representation of the symmetric group $\mathfrak{S}_n$. (Note that the triviality of the Yang-Baxter equation can be more easily seen with the R-matrix: if $R_{\g\d}^{\a\b}(\theta)$ is diagonal,
then $R$'s acting on different components commute (see \cite[Proposition 3.7]{BT15})
\footnote{R-matrices should not be confused with the residues $R_{\a\b}$ of $S^{\b\a}_{\a\b}$.
We do not use R-matrices in the rest of the paper.}.

We are in particular interested in the space of \textbf{$S$-symmetric functions} in $\mathcal{H}_1^{\otimes n}$,
namely functions which are invariant under this action of $\mathfrak{S}_n$:
those vector for which it holds for any $k$ that
\begin{align}\label{psisym}
  \Psi_n^{\pmb{\a}}(\pmb{\theta}) = S_{\a_{k+1}\a_k}^{\a_k\a_{k+1}}(\theta_{k+1}-\theta_k)\Psi^{\a_1\cdots\a_{k+1}\a_k\cdots\a_n}_n(\theta_1,\cdots, \theta_{k+1},\theta_k,\cdots,\theta_n).
\end{align}
With these functions we can define the Hilbert space $\mathcal{H}$ of the theory, first by defining
the $n$-particle Hilbert space $\mathcal{H}_n$ as the subspace of
$S$-symmetric functions in $\H_1^{\otimes n}$, and then by considering $\mathcal{H} := \oplus_{n=0}^\infty \mathcal{H}_n$ with
$\mathcal{H}_0 =\mathbb{C}\Omega$.
We introduce the orthogonal projection $P_n := \frac{1}{n!}\sum_{\sigma \in \mathfrak{S}_n} D_n(\sigma)$
thus we can write $\mathcal{H}_n = P_n \mathcal{H}_1 ^{\otimes n}$, and we denote with $\mathcal{D}$
the dense subspace of $\mathcal{H}$ of vectors with finite particle number.
The elements of $\mathcal{H}$ are sequences $\Psi = (\Psi_0, \Psi_1, \ldots)$, where $\Psi_n \in  P_n \mathcal{H}_1 ^{\otimes n}$
such that $\| \Psi \|^2 := \sum_{n=0}^\infty \| \Psi_n \|^2 < \infty$.

Now we discuss Poincar\'e symmetries on $\mathcal{H}$.
The representation $U_1$ of the Poincar\'e group $\poincare$ can be promoted to $\H$
by the second quantization, which acts as
\begin{align*}
(U(a,\lambda)\Psi)^{\pmb{\alpha}}_n(\pmb{\theta}) := \exp \left(i \sum_{l=1}^n p_{\alpha_l}(\theta_l) \cdot a\right) \Psi^{\pmb{\alpha}}_n(\pmb{\theta}-\pmb{\lambda}),
\end{align*}
where $\pmb{\lambda}=(\lambda, \ldots, \lambda)$ and $p_{\alpha_l}(\theta)= (m_{\alpha_l}\cosh \theta, m_{\alpha_l}\sinh \theta)$.
Additionally, the CPT operator extends to $\mathcal{H}$ by
\begin{align*}
(J \Psi)_n^{\pmb{\alpha}}(\pmb{\theta}) := \overline{\Psi_n^{\bar\alpha_n \ldots \bar\alpha_1} (\theta_n, \ldots, \theta_1)}.
\end{align*}
On $\mathcal{H}_1$, $J$ indeed reduces to \eqref{J1}.
It was shown in \cite[Lemma 2.3]{LS14} that $U(a,\lambda)$ and $J$ leave the space of $S$-symmetric functions $P_n\H_n$ invariant:
the proof does not use the analyticity of $S$, hence carries over to our cases.

We will deal with one-particle wave functions $g \in \bigoplus_{\alpha \in \I} \mathscr{S}(\mathbb{R}^2)$ with several components $g_{\alpha} \in \mathscr{S}(\mathbb{R}^2)$ and we will adopt the following convention \cite{LS14}:
\begin{equation*}
g^{\pm}_{\alpha}(\theta) := \frac{1}{2\pi} \int d^2 x\, g_\alpha(x) e^{\pm ip_\alpha(\theta)\cdot x}.
\end{equation*}
If $g_\a$ is supported in $W_\R$, then $g^+_\a(\theta)$ has a bounded analytic continuation in $\RR +i(-\pi, 0)$ and
$|g^+_\a(\theta +i\lambda)|$ decays rapidly as $\theta \rightarrow \pm \infty$ in the strip for $\l \in (-\pi,0)$.
Furthermore, the CPT operator $j$ acts naturally on multi-components test functions as $(g_j)_\alpha (x) := \overline{g_{\bar \alpha}(-x)}$,
and it holds that $((g_j)^\pm)_\alpha(\zeta) = \overline{g^\pm_{\bar\alpha}(\bar \zeta)}$.
Moreover, we recall $(g^*)_\alpha (x) :=\overline{g_{\bar\alpha}(x)}$ and say that $g$ is \textbf{real} if $g=g^*$
(as \cite[Proposition 3.1]{LS14}, which is a generalization of real-valuedness in the scalar case \cite{CT15-1}).
If $g$ is real, then $\overline{g^\pm_{\bar\alpha}(\overline \zeta)} = g^\mp_\alpha(\zeta)$.

There is a natural action of the proper Poincar\'e group on $\RR^2$ and on the space of test functions,
denoted by $g_{(a,\lambda)}$ (while the space-time reflection acts by $g \mapsto g_j$), and it is compatible with the action on the one-particle space:
\[
 (g_{(a, \lambda)})^\pm_\alpha = U_1(a,\lambda)g^\pm_\alpha,\;\;\;
 (g_j)^\pm_\alpha = J_1 g^\pm_\alpha.
\]

\subsubsection*{Creation and annihilation operators}

Similarly to \cite{LS14}, creators and annihilators $z^\dagger_\alpha (\theta), z_\alpha (\theta)$ are introduced
in the $S$-symmetric Fock space $\mathcal{H}$. Their actions on vectors $\Psi = (\Psi_n) \in \D$ are given by,
for $\varphi \in \H_1$,
\begin{align}
(z(\varphi)\Psi)^{\pmb{\a}}_n (\pmb{\theta}) &= \sqrt{n+1}\sum_\nu\int d\theta' \overline{\varphi^\nu(\theta')}\Psi^{\nu \pmb{\a}}_{n+1}(\theta',\pmb{\theta}),\\
z^\dagger(\varphi) &= (z(\varphi))^*
\end{align}
(see \cite[Proposition 2.4]{LS14})
and they formally fulfill the following Zamolodchikov-Faddeev algebraic relations:
\begin{align*}
z^\dagger_{\alpha}(\theta)z^\dagger_{\beta}(\theta') &= S^{\beta \alpha}_{\alpha \beta}(\theta -\theta')z^\dagger_\beta(\theta')z^\dagger_\alpha(\theta),\\
z_\alpha(\theta)z_\beta(\theta') &= S^{\beta \alpha}_{\alpha \beta}(\theta -\theta')z_\beta(\theta')z_\alpha(\theta),\\
z_\alpha(\theta)z^\dagger_\beta(\theta') &= S^{\alpha \beta}_{\beta \alpha}(\theta' -\theta)z^\dagger_\beta(\theta')z_\alpha(\theta)+\delta^{\alpha \beta}\delta(\theta -\theta')\pmb{1}_{\mathcal{H}}.
\end{align*}
Recall that they are defined on $\D$ and bounded on each $n$-particle space $\H_n$.

They can alternatively be defined in terms of the corresponding unsymmetrized creators and
annihilators $a(f), a^\dagger(f)$, $f \in \mathcal{H}_1$, by setting $z^{\#}(f):=Pa^{\#}(f)P$,
where $P:= \bigoplus_{n=0}^\infty P_n$ is the orthogonal projection from the unsymmetrized Fock
space to the $S$-symmetric Fock space $\mathcal{H}$ \cite[(2.23--26)]{LS14} and $\#$ stands for either $\dagger$ or nothing
(either creator or annihilator).

\subsubsection*{Wedge-local fields for analytic S-matrices}

For the class of two-particle S-matrices $S(\zeta)$ with components which are
\emph{analytic in the physical strip} $\zeta \in \mathbb{R} +i(0,\pi)$,
local observables associated with wedge-regions, say with the standard left wedge $W_{\mathrm L}$,
can be constructed \cite{LS14}, following an argument due to Schroer \cite{Schroer97}.
Specifically, let $f \in \bigoplus_{\a\in \I}\mathscr{S}(\mathbb{R}^2)$,
Lechner and Sch\"utzenhofer defined a multi-component quantum field $\phi$ by
\begin{align*}
\phi(f)&:= z^\dagger(f^+)+ z(J_1f^-) \\
       &\left(= \sum_\alpha \int d\theta\, \left( f^+_\alpha(\theta) z^\dagger_\alpha(\theta)+ (J_1f^-)_\alpha (\theta)z_\alpha(\theta) \right)\right).
\end{align*}
We note that this reduces to the free field if $S^{\alpha \beta}_{\beta \alpha}(\theta)=1$
(which however violates \ref{atzero} and therefore it is not in the class of scattering functions considered here).
The field $\phi$ shares many properties with the free field as shown by Lechner and Sch\"utzenhofer in
\cite[Proposition 3.1]{LS14}
In particular, it is defined on the subspace $\D$ of $\mathcal{H}$ of vectors with finite particle number and it is essentially
self-adjoint on $\D$ for test functions $f$ with the property that $f = f^*$ (we denote its closure by the same symbol $\phi(f)$).
It satisfies the Reeh-Schlieder property and transforms covariantly under the representation
$U(x,\lambda)$ of the proper orthochronous Poincar\'e group.
The only exception is the property of locality. The field $\phi(x)$ is not localized
at the space-time point $x$ in the usual sense, but rather in an infinitely extended wedge
with tip at $x$, $W_{\mathrm L} +x$. To make this more precise, one introduces the reflected creators and annihilators \cite{LS14}, as
\begin{align*}
z'_\alpha(\theta):= J z_{\bar\alpha}(\theta)J, \quad z'^\dagger_\alpha(\theta):= J z_{\bar\alpha}^\dagger(\theta)J
\end{align*}
and defines a new field $\phi'$ as, $f \in \mathscr{S}(\mathbb{R}^2)$,
\begin{equation*}
\phi'(f):= J \phi(f_j)J = z'^\dagger (f^+) + z'(J_1 f^-).
\end{equation*}
It has been shown in \cite[Theorem 3.2]{LS14} that the two fields $\phi,\phi'$
are relatively \emph{wedge-local}, in the sense  that the commutator $[e^{i\phi(f)}, e^{i\phi'(g)}]$
is zero for any test functions $f,g$ with the property that $f=f^*$ and $g=g^*$, and with $\supp f \subset W_\L$, $\supp g \subset W_\R$.
Hence, we can interpret $\phi, \phi'$ as observables measurable in the wedges $W_\L, W_\R$, respectively.
This result can be obtained by computing the commutators of $z^\#$ with $z'^\#$
as shown in \cite[Theorem 3.2]{LS14} and by shifting a certain integral contour
which critically uses the analyticity of the two-particle S-matrix $S(\theta)$
in the physical strip $\theta \in \mathbb{R} +i(0,\pi)$.

It should be remarked \cite[Theorem 3.2]{LS14} that also the properties of
the test functions $f,g$ play an important role in the proof of wedge-locality.
More specifically, the proof uses the fact that if $\supp f_\a \subset W_\L$
then its Fourier transform $f_\a^+$
is analytic, bounded in $\RR + i(0,\pi)$ and $|f_\a^+(\theta +i\lambda)|$ decays rapidly as
$\theta \rightarrow \pm \infty$ for $\lambda \in (0,\pi)$. Moreover, $f^+_\alpha(\theta +i\pi)= f^-_\alpha(\theta)$
holds always, and if $f$ satisfies the reality condition $f = f^*$,
we have in addition $\overline{f^+_{\bar\a}(\theta)} = f^-_\alpha(\theta) = f^+_\alpha(\theta +i\pi)$
(similar remarks apply to $g$ as well).

If the S-matrix is scalar without poles and satisfies a certain regularity condition,
then Lechner proved \cite{Lechner08} that
one can construct Haag-Kastler nets, in which the local algebras are nontrivial at least for
double cones larger than a certain minimal size \cite{Alazzawi14},
by showing the so-called modular nuclearity condition \cite{BL04}.
Some progress has been made for models with matrix-valued S-matrices without poles in the physical strip, and
a proof of modular nuclearity appears to be available for diagonal subcases \cite{Alazzawi14}.
Our goal is also to construct Haag-Kastler nets for S-matrices with poles, yet
in the present paper we will not investigate this strict locality and concentrate on wedge-local aspects.

For the class of two-particle S-matrices $S(\theta)$ with components which are
\emph{not analytic in the physical strip} $\theta \in \mathbb{R} +i(0,\pi)$,
the fields $\phi(f), \phi'(g)$ fail to be wedge-local (as seen already in the scalar case \cite{CT15-1}),
therefore, some modifications are necessary, as we will see in Section \ref{chi}.

\section{The bound state operator}\label{chi}
We are going to define a generalization of the operator $\chi(f)$ of \cite{CT15-1}.
Its mathematical structure is parallel to that in the case with scalar S-matrix, yet
this time the name ``bound state operator'' fits better, as it clearly corresponds to the fusion table of the model.

\subsection{Construction}

Let $f$ be a multi-component test function whose components are supported in $W_\L$.
Hardy spaces on strip appear as a very important ingredient (see also \cite{Tanimoto15-1}):
\[
H^2(\l_1,\l_2) = \{\xi: \xi \mbox{ is analytic on } \RR + i(\l_1,\l_2), \|\xi(\cdot + i\l)\|_{L^2(\RR,d\theta)} \le C_\xi, \l \in (\l_1,\l_2)\}.
\]
An element $\xi \in H^2(\l_1,\l_2)$ has an $L^2$-boundary value at $\l = \l_1,\l_2$ and we can consider
$H^2(\l_1,\l_2)$ as a subspace of $L^2(\RR)$ by choosing one of the boundaries.

Now we introduce an unbounded operator $\chi(f)$ on the $S$-symmetric Fock space $\H$.
Let us denote its component on $\H_n$ by $\chi_n(f)$.
Firstly, $\chi_0(f)$ annihilates the vacuum $\Om$.

Let us first assume that $f$ has only one non-zero component $f_\a$ with index $\a$, 
and moreover it has support in the left wedge $W_\L$.
The action $\chi_{1,\a}(f)$ on $\H_1$ is then given as follows:
\begin{align}
  \dom(\chi_{1,\a}(f)) &:= \bigoplus_\b \left\{\begin{array}{ll}
                                          H^2(-\theta_{(\b\a)},0) &\mbox{if } (\a\b) \mbox{ fuse into some } \g \\
                                          L^2(\RR) &\mbox{otherwise}
                                         \end{array}\right. \label{domain:chi1} \\
 (\chi_{1,\a}(f)\xi)^\gamma(\theta) &:= \left\{\begin{array}{ll}
                                           -i\eta^\gamma_{\alpha \beta} f^+_\alpha (\theta + i\theta_{(\alpha \beta)} ) \xi^\beta (\theta - i\theta_{(\beta \alpha)}) & \mbox{if }(\a\b) \mbox{ fuse into } \g\\
                                           0 &\mbox{otherwise}
                                          \end{array}\right. \label{eq:chi1}
\end{align}
where $\eta^\gamma_{\alpha \beta} := i\sqrt{2\pi |R_{\alpha \beta}^\g|}$,
$R_{\alpha \beta}^\g$ is given in \ref{polestructure}.
If $\a = \up$ is an elementary particle, then it follows from Section \ref{positive-residue}
that $|R^\gamma_{\alpha\beta}| = -iR^\gamma_{\alpha\beta}$ (actually this holds for
any index $\a$ in examples of Section \ref{examples}, yet we will use only those associated with
elementary particles).
Note that $R_{\a\b}^\g = 0$ by definition if $(\a\b)\to\g$ is not a fusion process.
Furthermore, since $f$ has support in $W_\L$,
$f^+_\a(\theta + i\theta_{(\alpha \beta)})$ is bounded (actually rapidly decreasing), and therefore $\chi_{1,\a}(f)\xi$ is $L^2$.
If $f$ has more than one non-zero components, we extend it by linearity:
$\chi_1(f) = \sum_\a\chi_{1,\a}(f_\a)$ on the natural domain (the intersection over $\a$).

If $(\a\b)\to\g$ is not a fusion process, the expression $\xi^\beta (\theta - i\theta_{(\beta \alpha)})$
does not make sense as $\xi^\b$ does not necessarily have an analytic continuation,
yet in such a case $\eta^\g_{\a\b} = 0$ by definition, hence
\begin{equation}\label{eq:formalchi}
  (\chi_1(f)\xi)_\gamma(\theta) = -\sum_{\a\b}i\eta^\gamma_{\alpha \beta} f^+_\alpha (\theta + i\theta_{(\alpha \beta)} ) \xi_\beta (\theta - i\theta_{(\beta \alpha)})
\end{equation}
holds in a formal sense. We make use of this formula when the domain question does not pose a serious problem.

We can interpret this one-particle action \eqref{eq:chi1} as the situation where
the state of one elementary particle $\xi^\beta$ is fused with $f^+_\alpha$ into a particle of type $\gamma$.
For this reason (which is clearer than the scalar case), we will call this operator the ``\textbf{bound state operator}''.

Then, we define:
\begin{align}\label{eq:defchi}
 \chi_n(f) &:= nP_n(\chi_1(f)\otimes\1\otimes\cdots\otimes\1)P_n, \\
 \chi(f) &= \bigoplus_{n=0}^\infty \chi_n(f). \nonumber
\end{align}
As $\chi(f)$ is the direct sum of $\chi_n(f)$, we define its domain to be the algebraic direct sum:
$\dom(\chi(f)) := \bigoplus_{n=0}^{\text{(finite)}}\dom(\chi_n(f))$, which is a subspace of $\D$.
As in \cite{CT15-1}, when a product $AB$ of possibly unbounded operators $A, B$ appears,
its domain is naturally understood as $\{\xi \in \dom(B): B\xi \in \dom(A)\}$.

Now, let $g$ be a test function supported in the right wedge $W_\R$,
we introduce the reflected bound state operator $\chi'(g)$:
again, for $g$ with only one non-zero component $g_\a$ we define
\begin{align*}
  \dom(\chi'_{1,\a}(g)) &:=\bigoplus_\b \left\{\begin{array}{ll}
                                          H^2(0,\theta_{(\beta \alpha)}) &\mbox{if } (\a\b) \mbox{ fuse into some } \g \\
                                          L^2(\RR) &\mbox{otherwise}
                                         \end{array}\right. \\
 (\chi'_{1,\a}(g)\xi)^\gamma(\theta) &:= \left\{\begin{array}{ll}
                                           -i\eta^\gamma_{\alpha \beta} g^+_\alpha\left(\theta - i\theta_{(\alpha \beta)}\right) \xi^\beta \left(\theta + i\theta_{(\beta \alpha)}\right) & \mbox{if }(\a\b) \mbox{ fuse into } \g\\
                                           0 &\mbox{otherwise}
                                          \end{array}\right.
\end{align*}
and $\chi'_1(g) = \sum_\a \chi'_{1,\a}(g_\a)$ for a general vector-valued test function $g = (g_\a)$.
Its component on $\H_n$ is given by
\begin{align*}
\chi'_n(g) := nP_n(\1\otimes\cdots\otimes\1\otimes\chi'_1(g))P_n
\end{align*}
and we define $\chi'(g) = \bigoplus_n \chi'_n(g)$.

This operator is related to $\chi$ by the CPT operator $J$:
\begin{equation}\label{chipfull}
 \chi'(g) = J\chi(g_j)J.
\end{equation}
Indeed, let us consider the one-particle components of this expression above.
We know that $J_1$ acts as complex conjugation composed with charge conjugation, and as such, it 
takes an analytic function in the lower strip to an analytic function in the upper strip with
an exchange in components $\a \leftrightarrow \bar \a$.
Because of the expression $J\chi(g_j)J$ in terms of $J$ and since $g_j$ has also charge-conjugated components,
we then have that the domains of $\chi'(g)$
and of $J\chi(g_j)J$ coincide.
One can also show that the operators do coincide by computing their actions on vectors
and noting that $(\a\b)\to\g$ is a fusion process if and only if $(\bar\a\bar\b)\to\bar\g$ is also a fusion
by assumption (see Section \ref{S-matrix}, fusion table).

\subsection{An alternative expression for \texorpdfstring{\eqref{eq:defchi}}{chi}}\label{alternative}
We have alternative expressions for $\chi(f)$ and $\chi'(g)$ corresponding to the scalar case \cite[Section 3.2]{CT15-1}.

Let $\tau_j \in \mathfrak{S}_n$, $1 \le j \le n-1$, be the transposition which exchanges $j$ and $j+1$, and let  $\rho_k =\tau_{k-1}\cdots\tau_1$ be the cyclic permutation
\[
 \rho_k: (1,2,\cdots, n) \mapsto (k,1,2, \cdots, k-1, k+1,\cdots, n).
\]
$\rho_1$ is the unit element of $\mathfrak{S}_n$ by convention. With this notation, since any permutation $\sigma \in \mathfrak{S}_n$ is a bijection of $\{1,\cdots, n\}$ onto itself, it can be written as the product $\rho_{\sigma(1)}\underline\sigma$
with a permutation $\underline{\sigma}$ of $n-1$ numbers $(2,3,\cdots, n)$.

By the definition of $\underline\sigma$, the operators $\chi_1(f)\otimes\1\otimes\cdots\otimes\1$ and $D_n(\underline{\sigma})$ commutes.
Moreover, one has that $D_n(\sigma)P_n = P_n$ since $P_n = \frac1{n!}\sum_{\sigma\in\mathfrak{S}_n} D_n(\sigma)$.
As in \cite{CT15-1}, we can rewrite the formula for $\chi_n(f)$ as follows:
\begin{align*}
 \chi_n(f) &= nP_n(\chi_1(f)\otimes\1\otimes\cdots\otimes\1)P_n \\
 &= \frac1{(n-1)!}\sum_{\sigma\in\mathfrak{S}_n} D_n(\rho_{\sigma(1)})D_n(\underline{\sigma})
 (\chi_1(f)\otimes\1\otimes\cdots\otimes\1)P_n \\
 &= \frac1{(n-1)!}\sum_{\sigma\in\mathfrak{S}_n} D_n(\rho_{\sigma(1)})
 (\chi_1(f)\otimes\1\otimes\cdots\otimes\1)P_n \\
 &= \sum_{1\le k \le n} D_n(\rho_k)(\chi_1(f)\otimes\1\otimes\cdots\otimes\1)P_n.
\end{align*}

We consider a vector $\Psi_n$ that is $S$-symmetric, i.e. $\Psi_n = P_n \Psi_n$,
and in the domain of $\chi_1(f)\otimes\1\otimes\cdots\otimes\1$.
We assume that $f$ has non-zero component $f_\a$.
Then, the components of $\Psi_n$ have a meromorphic continuation (due to the presence of $S$ factors) in variable $\theta_\ell$
for which $(\a\b_l)$ is a fusion. If $(\a\b_1)$ is a fusion, then it holds that
\begin{multline}
 \Psi_n^{\pmb{\b}}\left(\theta_1 - i\theta_{(\b_1\a)}, \theta_2, \cdots, \theta_n\right)
 \\
 = \prod_{2\le j \le k} S^{\b_1 \b_j}_{\b_j \b_1}\left(\theta_j-\theta_1 + i\theta_{(\b_1\a)}\right)
\Psi_n^{\b_2 \ldots \b_k \b_1 \b_{k+1}\ldots \b_n}\left(\theta_2,\cdots,\theta_k,\theta_1 - i\theta_{(\b_1\a)},\theta_{k+1},\cdots,\theta_n\right).
\end{multline}
Using this fact, and by recalling that $\eta^\g_{\a\b} \neq 0$ only if $(\a\b)\to\g$ is a fusion process, we can compute each term $D_n(\rho_k)(\chi_1(f)\otimes\1\otimes\cdots\otimes\1)P_n$ in the expression of $\chi_n(f)$ above, for $k \geq 2$ (the case $k=1$ is trivial by definition of $\rho_k$), as follows :
\begin{align*}
 &(D_n(\rho_k)(\chi_1(f)\otimes\1\otimes\cdots\otimes\1)\Psi_n)^{\gamma_1 \ldots \gamma_n}(\theta_1\cdots\theta_n) \\
 =\;& \prod_{1\le j \le k-1} S^{\gamma_j \gamma_k}_{\gamma_k \gamma_j}(\theta_k-\theta_j) \\
 &\quad \quad \quad \times ((\chi_1(f)\otimes\1\otimes\cdots\otimes\1)\Psi_n)^{\gamma_k \gamma_1 \ldots \gamma_{k-1}\gamma_{k+1} \ldots \gamma_n}(\theta_k,\theta_1,\cdots,\theta_{k-1},\theta_{k+1},\cdots,\theta_n) \\
 =\;& -\sum_{\beta_k \in \I}i\eta^{\gamma_k}_{\alpha \beta_k}\prod_{1\le j \le k-1} S^{\gamma_j \gamma_k}_{\gamma_k \gamma_j}(\theta_k-\theta_j)  f^+_\alpha\left(\theta_k + i\theta_{(\alpha \beta_k)}\right) \\
 &\quad \quad \quad \times \Psi_n^{\beta_k \gamma_1 \ldots \gamma_{k-1} \gamma_{k+1} \ldots \gamma_n}\left(\theta_k - i\theta_{(\beta_k \alpha)},\theta_1,\cdots,\theta_{k-1},\theta_{k+1},\cdots\theta_n\right) \\
 =\;& -\sum_{\beta_k \in \I}i\eta^{\gamma_k}_{\alpha \beta_k} \prod_{1\le j \le k-1} S^{\gamma_j \gamma_k}_{\gamma_k \gamma_j}(\theta_k-\theta_j)  S^{\beta_k \gamma_j}_{\gamma_j \beta_k}\left(\theta_j-\theta_k + i\theta_{(\beta_k \alpha)}\right)
 f^+_\alpha\left(\theta_k + i\theta_{(\alpha \beta_k)}\right) \\
 &\quad \quad \quad \times \Psi_n^{\gamma_1 \ldots \gamma_{k-1} \beta_k \gamma_{k+1} \ldots \gamma_n}\left(\theta_1,\cdots,\theta_{k-1},\theta_k -i\theta_{(\beta_k \alpha)},\theta_{k+1},\cdots\theta_n\right) \\
 =\;& -\sum_{\beta_k \in \I} i\eta^{\gamma_k}_{\alpha \beta_k} \prod_{1\le j \le k-1} S^{\gamma_j \alpha}_{\alpha \gamma_j}\left(\theta_k-\theta_j + i\theta_{(\alpha \beta_k)}\right)
 f^+_{\alpha}\left(\theta_k + i\theta_{(\alpha \beta_k)}\right)\\
 & \quad \quad \quad \times \Psi_n^{\gamma_1 \ldots \gamma_{k-1} \beta_k \gamma_{k+1} \ldots \gamma_n}\left(\theta_1,\cdots,\theta_{k-1},\theta_k - i\theta_{(\beta_k \alpha)}, \theta_{k+1},\cdots\theta_n\right),
\end{align*}
where in the second equality we used the fact that $\eta^\g_{\a\b}=0$ unless $(\a\b)\to\g$
and we simplified expression \eqref{eq:formalchi} (namely, those terms with $\b_k$ for which
$(\a\b_k)$ is not a fusion should be simply ignored),
in the third equality we reordered the variables,
and in the last equality we used the bootstrap equation, parity symmetry and hermitian analyticity.

Note that the $S$-factors appearing in the above expression have poles,
therefore $\Psi_n$ must have zeros at the location of these poles, so that the whole expression remains $L^2$. This is the meaning of $\Psi_n$ being in the domain of $\chi_n(f)$ (for details of this domain see Proposition~\ref{pr:chi:symmetry}.)

Similarly to $\chi(f)$, also $\chi'(g)$ can be written in the form:
\begin{align*}
 \chi'_n(g) &= \sum_{1\le k \le n} D_n(\rho'_k)(\1\otimes\cdots\otimes\1\otimes\chi'_1(g))P_n,
\end{align*}
where $\rho'_k = \tau_{n-k+1}\tau_{n-k+2}\cdots\tau_{n-1}$ are the cyclic permutations
\[
 \rho'_k:(1,\cdots, n-1,n) \longmapsto (1,\cdots n-k, n-k+2,\cdots, n-1, n, n-k+1)  
\]
and, for $k \ge 2$,
\begin{align}\label{chip}
 &(D_n(\rho'_k)(\1\otimes\cdots\otimes\1\otimes\chi'_1(g))\Psi_n)^{\gamma_1 \ldots \gamma_n}(\theta_1\cdots\theta_n)  \nonumber\\
 =&\; -\sum_{\beta_{n-k+1}\in\I}i\eta_{\alpha \beta_{n-k+1}}^{\gamma_{n-k+1}}\prod_{n-k+2 \le j \le n} S^{\alpha \gamma_j}_{\gamma_j \alpha}\left(\theta_j- \theta_{n-k+1} + i\theta_{(\alpha \beta_{n-k+1})}\right)
 g^+_\alpha\left(\theta_{n-k+1} - i\theta_{(\alpha \beta_{n-k+1})}\right) \nonumber \\
 &\quad\times (\Psi_n)^{\gamma_1 \ldots \gamma_{n-k}\beta_{n-k+1}\gamma_{n-k+2} \ldots \gamma_n}\left(\theta_1,\cdots,\theta_{n-k},\theta_{n-k+1} + i\theta_{(\beta_{n-k+1}\;\alpha)},\theta_{n-k+2},\cdots\theta_n\right).
\end{align}
This is valid for $k \ge 1$ (the case $k=1$ is trivial by definition of $\rho'_k$.)

\subsection{Some properties}
We show here some properties of $\chi(f)$ which are naturally expected and necessary for the further developments.
Note that since $\chi'(g)$ has a construction similar to $\chi(f)$, they share these properties,
so that it is enough to show them for $\chi(f)$.

\begin{proposition}\label{pr:chi:symmetry}
 For a vector-valued test function $f$ supported in $W_\L$ and with the property that $f^*=f$,
 the operator $\chi(f)$ is densely defined and symmetric.
\end{proposition}
\begin{proof}

We first consider the operator $\chi_1(f)$, where $f$ has only non-zero components with indices $\a$ and $\bar \a$.
The general case follows, as $\chi_1(f)$ is a linear combination of such operators, yet the domain is always dense.

We can write $(\chi_{1,\a}(f_\a) \xi)^\gamma = \sum_{\beta}\eta^\gamma_{\alpha \beta}M_{f_\alpha(\cdot+i\theta_{(\a\b)})}\Delta_1^{\theta_{(\beta \alpha)}/2\pi} \xi^\beta$
(recall our simplified notation \eqref{eq:formalchi}),
where $M_{f_\alpha(\cdot+i\theta_{(\a\b)})}$ is the multiplication operator by $f^+_\alpha(\theta +i\theta_{(\alpha \beta)})$
and $\Delta_1$ implements the analytic continuation $\left(\Delta_1^{\theta_{(\beta \alpha)}/2\pi} \xi^\beta \right)(\theta) = \xi^\beta (\theta -i\theta_{(\beta \alpha)})$.

As $\Delta_1^{it}$ implements the real shift $(\Delta_1^{it}\xi)(\theta) = \xi(\theta + 2\pi t)$,
$\Delta_1$ is a positive self-adjoint operator which is
unitarily equivalent to the multiplication operator by an exponential function through Fourier transform. Hence, its domain is dense in the Hilbert space.
Moreover, as $M_{f_\alpha(\cdot+i\theta_{(\a\b)})}$ is bounded, then $\chi_{1,\a}(f_\a)$ has a dense domain.
The linear combination $\chi_{1,\a}(f_\a) + \chi_{1,\bar \a}(f_{\bar\a})$ is also densely defined
as the intersection of $\dom(\Delta_1^\epsilon), \epsilon > 0$, is still dense.

To prove that $\chi_1(f)$ is symmetric, we take two vectors $\xi, \psi \in \dom(\chi_1(f))$ whose components have compact inverse Fourier transforms.
By the same argument as \cite[Proposition 3.1]{CT15-1}, such vectors form a core for $\chi_1(f)$.

Hence, it is enough to show that $\chi_1(f)$ is symmetric on a core, then it is symmetric on the whole domain. 
For $\xi, \psi$ as above, we compute 
 \begin{align}\label{chi1symm}
  \<\psi, \chi_{1,\a}(f)\xi\> &=  \sum_{\beta \gamma}\sqrt{2\pi |R_{\alpha \beta}^\g|} \int d\theta\, \overline{\psi_\gamma(\theta)} f^+_\alpha (\theta +i\theta_{(\alpha \beta)}) \xi^\beta (\theta -i\theta_{(\beta \alpha)}) \nonumber\\
  &= \sum_{\beta \gamma}\sqrt{2\pi |R_{\alpha \beta}^\g|} \int d\theta\, \overline{\psi_\gamma(\theta -i\theta_{(\beta \alpha)})} f^+_\alpha (\theta +i\theta_{\alpha \beta}) \xi^\beta (\theta)\nonumber\\
  &= \sum_{\beta \gamma}\sqrt{2\pi |R_{\alpha \beta}^\g|} \int d\theta\, \overline{\psi_\gamma(\theta -i\theta_{(\beta \alpha)}) f^+_{\bar\alpha}(\theta +i\pi -i\theta_{\alpha \beta})} \xi^\beta(\theta)\nonumber\\
  &=\sum_{\beta \gamma}\sqrt{2\pi |R_{\bar\alpha \gamma}^\b|} \int d\theta\, \overline{\psi_\gamma(\theta -i\theta_{(\gamma \bar\alpha)}) f^+_{\bar\alpha}(\theta +i\theta_{(\bar\alpha \gamma)})} \xi^\beta(\theta) \nonumber\\
  &= \<\chi_{1,\bar\a}(f_{\bar\a})\psi, \xi\>,
 \end{align}
where we used the Cauchy theorem in the second equality:
we can perform such shift of the integral contour since the integrand is analytic,
bounded and rapidly decreasing in the strip $\mathbb{R}+i(0,\pi)$.
In particular, $\xi,\psi$ are the Fourier transforms of compactly supported functions,
therefore they are rapidly decreasing, while $f^+$ is analytic and bounded in $\mathbb{R}+i(0,\pi)$
since $\supp f \subset W_\L$, and finally, we recall that $\overline{\psi_\gamma(\overline\zeta)}$ is analytic in $\zeta$.
We also used that $f^+_\alpha(\zeta) = \overline{f^-_{\bar\alpha}(\overline{\zeta})} = \overline{f^+_\alpha(\overline{\zeta} + i\pi)}$
in the third equality, and the properties of the residues and of the angles in
\ref{p:angles} and \ref{p:bootstrap} in the fourth equality.
By linearity in $f$, it follows that $\chi_1(f)$ is densely defined and symmetric.

Now we need to show that the same holds for $\chi_n(f)$. 
We start by showing that the domain of $\chi_n(f)$ is dense in $\H_n$.
Hence, we need to see that $P_n(\chi_1(f)\otimes\1\otimes\cdots\otimes \1)P_n$ is densely defined.
Since $P_n$ is bounded, it suffices to show that $(\chi_1(f)\otimes\1\otimes\cdots\otimes \1)P_n$ is densely defined.
The range of $P_n$ is the set of $S$-symmetric vectors, and the domain of
$\chi_{1,\a}(f_\a)\otimes\1\otimes\cdots\otimes \1$ is the vectors which have an $L^2$-bounded analytic continuation
to $-i\theta_{(\beta \alpha)}$ in the first variable with the $\beta$-component such that the fusion $(\alpha \beta)$ exists.
 
For an arbitrary set $\{\xi_1, \cdots \xi_n\}$ of $n$ vectors in $\dom(\chi_1(f))$,
$P_n(\xi_1\otimes \cdots \otimes \xi_n)$
is a $S$-symmetric vector, but have poles which come from the $S$-factors.
Since each $\xi_k$ is in $\dom(\chi_1(f))$, linear combinations of the above vectors form a dense subspace of $\H_n = P_n\H^{\otimes n}$.

In order to compensate these poles, we can multiply the vector $(n-2)!$ times by an extra factor $C_n(\pmb{\theta})$, which is given as
\[
  C_n(\pmb{\theta}) := \prod_{1 \le k < j \le n} \prod_{p} 
  \frac{(\theta_j - \theta_k - i \l_p)
  (\theta_k - \theta_j - i\l_p)}{
  (\theta_j - \theta_k - i\l_p')
  (\theta_k - \theta_j - i\l_p')},
\]
where the second product runs over all poles $\{i\l_p\}$ of all components of $S$ including multiplicity
in the physical strip $\strip$ (which is a finite set) and $\{i\l'_p\}$ are complex numbers in
the lower strip $\RR+i(-\pi,0)$.

This function is symmetric on $\mathbb{C}$, it is bounded in the physical strip and invertible on the real line.
Therefore, the multiplication operator $M_{C_n}$ by $C_n$ on all the components at the same time
preserves the $S$-symmetry of the vector, moreover it is still $L^2$
(because of the multiplication with a bounded function), and hence in $\H_n$,
and the invertibility of $M_{C_n}$ guarantees that this operator maps a dense subspace into another dense subspace.

After multiplication with $C_n(\pmb{\theta})$, the vector is now analytic in the first
 (actually any) variable to $i\theta_{(\beta \alpha)}$ and $S$-symmetric, and therefore, it is in the domain of
 $(\chi_1(f)\otimes\1\otimes\cdots\otimes\1)P_n$.
This proves that linear combinations of vectors of the form $(M_{C_n} P_n(\xi_1\otimes \cdots \otimes \xi_n))(\theta_1,\cdots,\theta_n)$ form a dense domain for $(\chi_1(f)\otimes\1\otimes\cdots\otimes \1)P_n$.
 
To prove that $P_n(\chi_1(f)\otimes\1\otimes\cdots\otimes \1)P_n$ is symmetric,  we consider two $n$-particle $S$-symmetric vectors $\Psi_n, \Phi_n$  in the domain of $\chi_n(f)$, and show that $\langle A \Psi_n, \Phi_n \rangle = \langle \Psi_n, A \Phi_n \rangle$, where $A$ is the linear operator above. Since $\Psi_n, \Phi_n$ are already $S$-symmetric, we only need to show that $\langle (\chi_1(f)\otimes\1\otimes\cdots\otimes \1) \Psi_n, \Phi_n \rangle = \langle \Psi_n, (\chi_1(f)\otimes\1\otimes\cdots\otimes \1) \Phi_n \rangle$, but this follows immediately from the computation in Eq.~\eqref{chi1symm} in the case of $\chi_1(f)$.

Finally, the  operator $\chi(f)$ is the direct sum of the $\chi_n(f)$, and its domain is defined as the finite algebraic direct sum of $\dom(\chi_n(f))$. Then, one can show that if $\chi_n(f)$ is densely defined and symmetric, the same holds for $\chi(f)$.

\end{proof}

We want to check now that $\chi(f)$ is covariant with respect to the action $U$ of the Poincar\'e group on $\H$ that we introduced in Section \ref{zf}. 

\begin{proposition}\label{pr:chi:covariance}
 Let $f$ be a test function supported in $W_\L$ and $(a,\lambda) \in \poincare$ such that
 $a \in W_\L$. Then it holds that
 $\Ad U(a,\lambda)(\chi(f)) \subset \chi(f_{(a,\lambda)})$.
\end{proposition}

\begin{proof}
 The proof of covariance is almost parallel to the scalar case \cite[Proposition 3.2]{CT15-1} except for translations, hence we will be brief.
 The proof can be restricted to the $n$-particle components. As $P_n$ commutes with $U(a,\l)$
 and $\chi_n(f) = P_n(\chi_1(f)\otimes\1\otimes\cdots\otimes\1)P_n$,
 it is enough to show the covariance of $\chi_1(f)$. The boost can be treated exactly as in the scalar case.
  
 A pure translation, $U_1(a,0)$ with $a \in W_\L$, requires the structure of fusion angles.
 By construction we have $(U_1(a,0)f^+)_\alpha(\theta) = e^{ia\cdot p_\alpha(\theta)}f^+_\alpha(\theta) = (f_{(a,0)})^+_\alpha(\theta)$.
 Hence $U_1(a,0)^*$ acts by multiplying with an exponential factor $e^{-ia\cdot p(\theta)}$
 and this factor has a bounded analytic continuation in $\RR + i(-\pi,0)$ for $a \in W_\L$, then $\dom(\chi_1(f)) = \dom(\chi_1(f_{(a,0)}))$.
 Let $\xi \in \dom(\chi_1(f))$, one checks that $\chi(f)$ is also covariant with respect to translations with the following computation 
(recalling always the simplified notation \eqref{eq:formalchi}):
 \begin{align*}
 &(U_1(a,0)\chi_1(f)U_1(a,0)^*\xi)^\gamma(\theta) \\
 &= e^{ia\cdot p_\gamma(\theta)}(\chi_1(f)U_1(a,0)^*\xi)^\gamma(\theta) \\
 &= \sum_{\alpha \beta} \sqrt{2\pi |R_{\alpha \beta}^\g|} e^{ia\cdot p_\gamma(\theta)}f^+_\alpha\left(\theta + i\theta_{(\alpha \beta)}\right)\cdot
 (U_1(a,0)^*\xi)^\beta\left(\theta - i\theta_{(\beta \alpha)}\right) \\
 &= \sum_{\alpha \beta} \sqrt{2\pi |R_{\alpha \beta}^\g|} e^{ia\cdot p_\gamma(\theta)} f^+_\alpha\left(\theta + i\theta_{(\alpha \beta)}\right)\cdot
 e^{-ia\cdot p_\beta \left(\theta - i\theta_{(\beta \alpha)}\right)}\xi^\beta\left(\theta - i\theta_{(\beta \alpha)}\right).
 \end{align*}
By Eq.~\eqref{prelations} we have $p_\gamma(\theta) - p_\beta\left(\theta - i\theta_{(\beta \alpha)}\right) =  p_\alpha\left(\theta + i\theta_{(\alpha \beta)}\right)$, hence
 \begin{align*}
 &(U_1(a,0)\chi_1(f)U_1(a,0)^*\xi)^\gamma(\theta) \\
 &= \sum_{\alpha \beta} \sqrt{2\pi |R_{\alpha \beta}^\g|} e^{ia\cdot \left(p_\gamma(\theta)- p_\beta\left(\theta -i\theta_{(\beta \alpha)}\right)\right)} f^+_\alpha\left(\theta + i\theta_{(\alpha \beta)}\right)\cdot
 \xi^\beta\left(\theta - i\theta_{(\beta \alpha)}\right) \\
 &= \sum_{\alpha\beta} \sqrt{2\pi |R_{\alpha \beta}^\g|} e^{ia\cdot p_\alpha\left(\theta + i\theta_{(\alpha \beta)}\right)} f^+_\alpha\left(\theta +i\theta_{(\alpha \beta)}\right)\cdot
 \xi^\beta\left(\theta - i\theta_{(\beta \alpha)}\right) \\
 &= (\chi_1(f_{(a,0)})\xi)^\gamma(\theta).
 \end{align*}
\end{proof}

\subsubsection*{Formal expression}
As in the scalar case \cite[Section 3.2]{CT15-1}, we can formally write down $\chi(f)$ in terms of $z^\dagger$ and $z$ as
\begin{align*}
\chi(f) &= \sum_{\alpha \beta \gamma}\sqrt{2\pi |R_{\alpha \beta}^\g|} \int d\theta\, f^+_\alpha\left(\theta +i\theta_{(\alpha \beta)}\right) z^\dagger_\gamma(\theta)z_\beta\left(\theta - i\theta_{(\beta \alpha)}\right), \\
\chi'(g) &= \sum_{\alpha \beta \gamma} \sqrt{2\pi |R_{\alpha \beta}^\g|} \int d\theta\, g^+_\alpha\left(\theta -i\theta_{(\alpha \beta)}\right) z'^\dagger_\gamma(\theta)z'_\beta\left(\theta + i\theta_{(\beta \alpha)}\right).
\end{align*}
Although these expressions look quite simple and attractive, we will not make use of it later in proofs,
as we do not have control on their operator domain. We also omit a formal justification of these expressions,
as it is parallel to that of the scalar case, one should only note again that $R_{\a\b}^\g = 0$ if
$(\a\b)\to\g$ is not a fusion process.

\section{The wedge-local fields}\label{wedge-local}

Similarly to \cite{CT15-1}, we introduce a field 
\begin{equation*}
\fct(f) = \phi(f) + \chi(f).
\end{equation*}
Since the domain of $\chi(f)$ contains vectors with finite particle number and with certain analyticity and boundedness properties (see the beginning of Sec.~\ref{chi}), its domain is included in the domain of $\phi(f)$, and therefore $\dom(\fct(f)) = \dom(\chi(f))$.

We also introduce the reflected field
\[
 \fct'(g) := \phi'(g) + \chi'(g) = J\fct(g_j)J.
\]
\begin{proposition}
Let $f$ be a test function with support in $W_\L$ and such that $f^*=f$, then $\fct(f)$ fulfills the following properties (similar results also holds for the reflected field $\fct'(g)$):
 \begin{enumerate}[{(}1{)}]
  \item $\fct(f)$ is symmetric. \label{pr:fct:symmetry}
  \item Let $f$ be a test function with components $f_\alpha = 0$ for all $\alpha \in \I$ except for some $\alpha_0 \in \I$, then $\fct$ is a solution of the Klein-Gordon equation with mass $m_{\alpha_0}$, \label{pr:fct:kg}
\begin{equation*}
\fct((\Box +m_{\alpha_0}^2)f)=0.
\end{equation*}
  \item $\fct(f)$ transforms covariantly under $U$, that is, if $f$ is supported in $W_\L$ and for $(a,\lambda) \in \poincare$ with $a \in W_\L$, we have $U(g)\fct(f)U(g)^* \subset \fct(f_{(a,\lambda)})$. \label{pr:fct:covariance}
\end{enumerate}
\end{proposition}

\begin{proof}
(\ref{pr:fct:symmetry}) In Proposition~\ref{pr:chi:symmetry} we showed that the operator $\chi(f)$ is symmetric, and we have from \cite[Proposition 3.1]{LS14} that $\phi(f)$ is also symmetric. Therefore, the sum $\fct(f)= \phi(f) + \chi(f)$ is symmetric.

(\ref{pr:fct:kg}) As mentioned in the statement of this proposition, $\fct$ is a solution of the Klein-Gordon equation if, for a test function $f$ with components $f_\alpha = 0$ for all $\alpha \in \I$ except for some $\alpha_0 \in \I$, one has $\fct((\Box +m_{\alpha_0}^2)f)=0$.
From \cite[P.14]{LS14} we know that $\phi((\Box + m_{\alpha_0}^2)f)=0$; on the other hand, $\chi((\Box +m_{\alpha_0}^2)f)$ acts by multiplication with $((\Box + m_{\alpha_0}^2)f)^+$ and $((\Box + m_{\alpha_0}^2)f)^+ =0$. So, the Klein-Gordon equation is indeed fulfilled.

(\ref{pr:fct:covariance}) In Proposition~\ref{pr:chi:covariance} we showed that $\chi(f)$ is covariant with respect to translations and boosts, and from \cite[Proposition 3.1]{LS14} we also know that $\phi(f)$ transforms covariantly under the proper orthochronous Poincar\'e group. Therefore, the sum $\chi(f) = \phi(f) + \chi(f)$ is covariant.
\end{proof}

As in the scalar case \cite{CT15-1}, the operator $\fct(f)$ has subtle domain properties,
in particular after applying this operator to a vector (not the vacuum) in its domain,
the $\phi(f)$ component generates a vector which by $S$-symmetry has further poles corresponding to those of S,
and therefore lies outside the domain of $\fct'(g)$.
Along with these subtleties, also the question whether the operator $\fct(f)$ is self-adjoint or
whether it admits self-adjoint extensions is still open.

\subsubsection*{Weak commutativity}

Because of the subtle domain property of $\fct(f)$ mentioned above,
we cannot form products of the type $\fct(f)\fct'(g)$ or $\fct'(g)\fct(f)$,
and compute the commutator $[\fct(f),\fct'(g)]$. We will need to evaluate the commutator 
on arbitrary vectors $\Phi,\Psi$ from a suitable space (see below),
and show $\<\fct(f)\Phi, \fct'(g)\Psi\> = \<\fct'(g)\Phi, \fct(f)\Psi\>$ for real $f, g$;
that is, we will show a weak form of commutativity on a domain for $f$ and $g$ associated
with the elementary particle $\up$ and its conjugate $\bar \up$.

For the scalar case, we conjectured that one can take simply $\dom(\fct(f))\cap\dom(\fct'(g))$.
In such a case, there is a hope to show that $\fct(g)$ and $\fct'(g)$ strongly commute using the argument of \cite{DF77}
(see \cite{Tanimoto15-1, Tanimoto15-2}).
In the diagonal case, the domain on which we can show the weak commutativity is strictly smaller than the simple intersection,
which poses a subtler problem. Fortunately, for models with two species of particles such as the $Z(3)$-Ising model
and the $A_2$-affine Toda theory, we can choose $\dom(\fct(f))\cap\dom(\fct'(g))$ with the same conjecture
as that in \cite[Section 3.3]{CT15-1}.

Let us start with studying some properties of vectors in $\dom(\fct(f))\cap\dom(\fct'(g))$.
Recall that we take $f$ and $g$ whose only nonzero components correspond to $\up, \bar\up$.
If $\Psi \in \dom(\fct(f))\cap\dom(\fct'(g))$, then all of its $n$-particle components $\Psi_n^{\pmb \a}$
have meromorphic continuations in the \emph{first} variable $\theta_1$ to $\theta_1 \pm i\theta_{(\a\up)}$,
as any index $\a$ can be fused either with $\up$ or $\bar \up$.
Then, due to $S$-symmetry, it has a meromorphic continuation in any variable with possible poles at the poles of $S$.
Moreover, for components $\Psi_n^{\pmb \a}$ with coinciding indices $\a_j=\a_k =: \b$,
we can exploit \ref{atzero}, $S_{\b\b}^{\b\b}(0)=-1$, as we did in \cite{CT15-1}:
we have
\begin{align*}
 &\Psi_n^{\a_1\cdots \b\cdots\b\cdots\a_n}(\theta_1,\cdots,\theta_j,\cdots,\theta_k,\cdots,\theta_n)\\
 &=\Psi_n^{\a_1\cdots \a_j\cdots\a_k\cdots\a_n}(\theta_1,\cdots,\theta_j,\cdots,\theta_k,\cdots,\theta_n)\\
 &= \Big(  \prod_{\ell = j+1}^{k-1} S^{\beta \alpha_{\ell}}_{\alpha_{\ell} \beta} (\theta_{\ell}-\theta_j) \Big) S_{\b\b}^{\b\b}(\theta_k-\theta_j) \Big( \prod_{\ell' = j+1}^{k-1} S^{\alpha_{\ell'} \beta}_{\beta \alpha_{\ell'}}(\theta_{k} - \theta_{\ell'})  \Big) \\
&\quad \times  \Psi_n^{\a_1\cdots \a_k\cdots\a_j\cdots\a_n}(\theta_1,\cdots,\theta_k,\cdots,\theta_j,\cdots,\theta_n)\\
 &= \Big(  \prod_{\ell = j+1}^{k-1} S^{\beta \alpha_{\ell}}_{\alpha_{\ell} \beta} (\theta_{\ell}-\theta_j)S^{\alpha_{\ell} \beta}_{\beta \alpha_{\ell}}(\theta_{k} - \theta_{\ell}) \Big)   S_{\b\b}^{\b\b}(\theta_k-\theta_j) \\
&\quad \times  \Psi_n^{\a_1\cdots \b\cdots\b\cdots\a_n}(\theta_1,\cdots,\theta_k,\cdots,\theta_j,\cdots,\theta_n)
\end{align*}
and therefore $\Psi_n^{\pmb \a}$ has a zero at $\theta_j - \theta_k = 0$. Note that for $\theta_j - \theta_k =0$, only the factor $S^{\b \b}_{\b \b}(\theta_k -\theta_j)$ remains in the above expression, while the other $S$-factors cancel due to \ref{hermitian}. However, we cannot infer the existence of zeros for other components and variables.

We will consider vectors from the following space:
\begin{equation}\label{domain}
\D_0 := \left\{ \Psi \in \D: \begin{array}{l}
                              \Psi_n^{\pmb \a} \text{ is analytic in } \mathbb{R}^n +i[-\theta_{(\b\up)}, \theta_{(\b\up)}]^n, \\
                              \Psi_n^{\pmb \a}(\pmb\theta + i\pmb \l) \in L^2(\RR^n) \text{ for } \pmb \l \in [-\theta_{(\b\up)}, \theta_{(\b\up)}]^n, \\
                              \text{ and has a zero at } \theta_i -\theta_j=0 \text{ for all } i,j 
                             \end{array}
        \right\},
\end{equation}
where $\theta_{(\b\up)}$ does not depend on $\b$, hence it is determined by the model.
Note that $\D_0 \subset \dom(\fct(f)) \cap \dom(\fct'(g))$.

The reason why we need zeros on the real plane is that in the computation of the weak commutator
the poles of the $S$-factors at $\theta_i -\theta_j =0$ appear from the shifting of integral contours
in expressions  such as below
\begin{equation*}
\prod_{1\le j,k \le n} S^{\alpha \beta}_{\beta \alpha}\left(\theta_j-\theta_k + i\theta_{\alpha \beta}\right)
 (\Psi_n)^{\gamma_1 \ldots \gamma_n}(\theta_1,\cdots,\theta_n)
\end{equation*}
and they must be canceled by the zeros of $\Psi_n$.
If we take $\Psi_n$ which has a zero at $\theta_j-\theta_k = 0$ which compensates
the pole of $S$, one can prove that the expression above remains
$L^2$ on the real line by \ref{regularity} \cite[Proposition E.7]{Tanimoto15-2}.

Using a similar argument as in the proof of Proposition \ref{pr:chi:symmetry}, one can
show that $\D_0$ is dense in $\mathcal{H}$. Namely, we can take the domain of Prop.\! \ref{pr:chi:symmetry}
and multiply each component further by
\[
 \prod_{1\le j<k\le n} \frac{(\theta_j-\theta_k)^2}{(\theta_j-\theta_k + i\l)(\theta_k-\theta_j + i\l)},
\]
where $|\l| > \theta_{(\b\up)}$, which yields a bounded $S$-symmetry-preserving operator with a dense range,
hence its image of the domain of Proposition \ref{pr:chi:symmetry} is again dense in each $\H_n$.

For vectors $\Psi,\Phi\in\D_0$, we can now prove the following theorem, which is our main result.

\begin{theorem}\label{theo:commutator}
Let $f$ and $g$ be test functions supported in $W_\L$ and $W_\R$, respectively, and with the property that $f=f^*$ and $g=g^*$.
Furthermore, assume that $f,g$ have components $f_\alpha =0$ and $g_\alpha =0$ for all $\alpha \in \I$
except the indices $\up,\bar\up$ corresponding to the elementary particles.

Then, for each $\Phi, \Psi$ in $\D_0$, we have
 \[
 \<\fct(f)\Phi, \fct'(g)\Psi\> = \<\fct'(g)\Phi, \fct(f)\Psi\>.
 \]
\end{theorem}

\begin{proof}
We should keep in mind that we are assuming that the vectors $\Phi$ and $\Psi$ are already $S$-symmetric.
Furthermore, we recall:
 \begin{align*}
  \fct(f) &= \phi(f) + \chi(f) = z^\dagger(f^+) + \chi(f) + z(J_1f^-), \\
  \fct'(g) &= \phi'(g) + \chi'(g) = z'^\dagger(g^+) + \chi'(g) + z'(J_1g^-).
 \end{align*}
Therefore, the (weak) commutator $[ \fct(f), \fct'(g)]$ expands into several terms that we will compute individually.

\begin{flushleft} 
 {\it The commutator $[\chi(f), z'(J_1g^-)]$}
\end{flushleft}
We can compute this commutator in the strong form, as there arises no problem of domains
(c.f.\! \cite[Theorem 3.4]{CT15-1}).
 
We recall the action of $z'$ for all $\varphi \in\H_1$,
\begin{align*}
\left( z'(J_1\varphi) \Psi\right)_n^{\pmb{\gamma}}(\pmb{\theta}) &= \left(Jz(J_1 \varphi)J \Psi \right)_n^{\pmb{\gamma}}(\pmb{\theta}) \\
        &= \overline{\left( z(J_1\varphi) J\Psi\right)_n^{\bar{\gamma_n} \ldots \bar{\gamma_1}}(\theta_n, \ldots, \theta_1)}\\
	&= \sqrt{n+1}\sum_{\b}\int d\theta'\, \overline{\varphi^{\bar \beta}(\theta')(J\Psi)^{\beta \bar{\gamma_n} \ldots \bar{\gamma_1}}_{n+1}(\theta',\theta_n,\ldots,\theta_1)}\\
	&=\sqrt{n+1}\sum_{\b}\int d\theta'\, \overline{\varphi^{\bar\beta}(\theta')} \Psi_{n+1}^{\pmb{\gamma}\bar\beta}(\pmb{\theta},\theta')\\
&=\sqrt{n+1}\sum_\b\int d\theta'\, \overline{\varphi^{\beta}(\theta')} \Psi_{n+1}^{\pmb{\gamma}\beta}(\pmb{\theta},\theta'),
\end{align*}
where we used the antilinearity of $z$ in the third equality and we renamed $\bar \beta$ into $\beta$ in the last equality.

Hence, using the alternative expression of $\chi(f)$ in Section \ref{alternative} we compute
\begin{align*}
  & (\chi(f)z'(J_1g^-)\Psi_n)^{\gamma_1 \ldots \gamma_{n-1}}(\theta_1,\cdots,\theta_{n-1}) \\
  &=-\sum_{k=1}^{n-1}\sum_{\alpha \kappa_k} i\eta^{\gamma_k}_{\alpha \kappa_k} f^+_\alpha (\theta_k +i\theta_{(\alpha \kappa_k)}) (Jz(J_1g^-)J\Psi_n)^{\gamma_1 \ldots \kappa_k \ldots \gamma_{n-1}}(\theta_1, \ldots, \theta_k -i\theta_{(\kappa_k \alpha)}, \ldots, \theta_{n-1})\\
&\;\; \times \left( \prod_{j=1}^{k-1}S^{\gamma_j \alpha}_{\alpha \gamma_j}(\theta_k -\theta_j +i\theta_{(\alpha \kappa_k)})\right)\\
&=- \sqrt{n} \sum_{k=1}^{n-1} \sum_{\alpha \kappa_k \beta} i\eta^{\gamma_k}_{\alpha \kappa_k} f^+_\alpha (\theta_k +i\theta_{(\alpha \kappa_k)}) \\
&\;\; \times \int d\theta'\, g^-_{\bar\beta}(\theta') (\Psi_{n})^{\gamma_1 \ldots \kappa_k \ldots \gamma_{n-1} \beta}(\theta_1 \ldots \theta_k -i\theta_{(\kappa_k \alpha)} \ldots \theta_{n-1},\theta')\left( \prod_{j=1}^{k-1}S^{\gamma_j \alpha}_{\alpha \gamma_j}(\theta_k -\theta_j +i\theta_{(\alpha \kappa_k)})\right).
\end{align*}
Similarly, the second term in the commutator gives
\begin{align*}
  & (z'(J_1g^-)\chi(f)\Psi_n)^{\gamma_1 \ldots \gamma_{n-1}}(\theta_1,\cdots,\theta_{n-1}) \\
&=\sqrt{n}\sum_\beta \int d\theta'\, g^-_{\bar\beta}(\theta') (\chi(f)\Psi_n)^{\gamma_1 \ldots \gamma_{n-1}\beta}(\theta_1,\ldots ,\theta_{n-1}, \theta')\\
&=-\sqrt{n}\sum_{k=1}^{n-1}\sum_{\alpha \kappa_k \beta} i\eta^{\gamma_k}_{\alpha \kappa_k} f^+_\alpha (\theta_k +i\theta_{(\alpha \kappa_k)}) \\
&\quad \times \int d\theta'\, g^-_{\bar\beta}(\theta')(\Psi_n)^{\gamma_1 \ldots \kappa_k \ldots \gamma_{n-1} \beta}(\theta_1 \ldots \theta_k -i\theta_{(\kappa_k \alpha)} \ldots \theta_{n-1} ,\theta')\left( \prod_{j=1}^{k-1}S^{\gamma_j \alpha}_{\alpha \gamma_j}(\theta_k -\theta_j +i\theta_{(\alpha \kappa_k)})\right)\\
&\quad -\sqrt{n}\sum_{\alpha \kappa \beta} i\eta^{\beta}_{\alpha \kappa} \int d\theta'\, g^-_{\bar\beta}(\theta') f^+_\alpha (\theta' +i\theta_{(\alpha \kappa)})
(\Psi_n)^{\gamma_1 \ldots \gamma_{n-1} \kappa}(\theta_1  \ldots \theta_{n-1},\theta' -i\theta_{(\kappa \alpha)})\\
&\qquad \times \left( \prod_{j=1}^{n-1}S^{\gamma_j \alpha}_{\alpha \gamma_j}(\theta' -\theta_j +i\theta_{(\alpha \kappa)})\right).
 \end{align*}
Combining the two above expressions, we find that $2\times(n-1)$ terms cancel each other and
\begin{align*}
&([\chi(f),z'(J_1g^-)]\Psi_n)^{\gamma_1 \ldots \gamma_{n-1}}(\theta_1,\cdots,\theta_{n-1})\nonumber\\
&\;=\;\sqrt{n} \sum_{\alpha \kappa \beta} i\eta^{\beta}_{\alpha \kappa} \int d\theta'\, g^-_{\bar\beta}(\theta') f^+_\alpha (\theta' +i\theta_{(\alpha \kappa)})
(\Psi_n)^{\gamma_1 \ldots \gamma_{n-1} \kappa}(\theta_1  \ldots \theta_{n-1},\theta' -i\theta_{(\kappa \alpha)})\nonumber\\
&\qquad \times \left( \prod_{j=1}^{n-1}S^{\gamma_j \alpha}_{\alpha \gamma_j}(\theta' -\theta_j +i\theta_{(\alpha \kappa)})\right) \nonumber\\
&\; =\;\sqrt{n} \sum_{\alpha \kappa \beta} i\eta^{\beta}_{\alpha \kappa} \int d\theta'\, g^-_{\bar\beta}(\theta') f^+_\alpha (\theta' +i\theta_{(\alpha \kappa)})
(\Psi_n)^{\kappa \gamma_1 \ldots \gamma_{n-1}}(\theta' -i\theta_{(\kappa \alpha)},\theta_1  \ldots \theta_{n-1})\nonumber\\
&\qquad \times  \left( \prod_{k=1}^{n-1}S^{\gamma_k \kappa}_{\kappa \gamma_k}(\theta' -\theta_k -i\theta_{(\kappa \alpha)})\right)\left( \prod_{j=1}^{n-1}S^{\gamma_j \alpha}_{\alpha \gamma_j}(\theta' -\theta_j +i\theta_{(\alpha \kappa)})\right) \nonumber\\
&\; =\; \sqrt{n} \sum_{\alpha \kappa \beta} i\eta^{\beta}_{\alpha \kappa} \int d\theta'\, g^-_{\bar\beta}(\theta') f^+_\alpha (\theta' +i\theta_{(\alpha \kappa)})
(\Psi_n)^{\kappa \gamma_1 \ldots \gamma_{n-1}}(\theta' -i\theta_{(\kappa \alpha)},\theta_1  \ldots \theta_{n-1})\nonumber\\
&\qquad \times  \left( \prod_{j=1}^{n-1}S^{\gamma_j \beta}_{\beta \gamma_j}(\theta' -\theta_j )\right), \\
\end{align*}
where in the second equality we used $S$-symmetry of $\Psi_n$ and in the third equality the Bootstrap equation.
Now we fix the indices of $f$ and $g$ to those $\up,\bar\up$ of elementary particles.
In the expression above, $\eta^\beta_{\alpha\kappa}$ is non-zero if and only if
$\alpha = \up, \beta = \bar\up$ or $\alpha = \bar\up, \beta = \up$ (because $(\up\k)\to\up$ etc.\! is impossible
by \ref{p:table}).
If $\eta^{\up}_{\bar\up\kappa} \neq 0$, then there is a fusion process $(\up\up)\to\kappa$
and by assumption $\kappa$ is the unique index for which
$S^{\kappa \up}_{\up \kappa}(\zeta)$ has a simple pole in $\RR+i[0,\theta_{(\kappa\up)}]$,
where $\theta_{(\kappa\up)}$ is independent of $\kappa$. Now, by the properties of $\dom(\fct(f))\cap\dom(\fct'(g)))$
we can continue as follows:
\begin{align}\label{chizp}
 &\; =\; \sqrt{n} \underset{\kappa\in\I}{\sum_{\nu = \up,\bar\up}} i\eta^{\nu}_{\bar\nu \kappa} \int d\theta'\, g^-_{\bar\nu}(\theta') f^+_{\bar\nu} (\theta' +i\theta_{(\bar\nu \kappa)})
 (\Psi_n)^{\kappa \gamma_1 \ldots \gamma_{n-1}}(\theta' -i\theta_{(\kappa \bar\nu)},\theta_1  \ldots \theta_{n-1})\nonumber\\
 &\qquad \times  \left( \prod_{j=1}^{n-1}S^{\gamma_j \nu}_{\nu \gamma_j}(\theta' -\theta_j )\right). \nonumber \\
&\; =\; \sqrt{n} \underset{\kappa\in\I}{\sum_{\nu = \up,\bar\up}} i\eta^{\nu}_{\bar\nu \kappa}
\int d\theta'\, g^-_{\bar\nu}(\theta' +i\theta_{(\kappa \bar\nu)}) f^+_{\bar\nu} (\theta' +i\theta_{\bar\nu \kappa})
(\Psi_n)^{\kappa \gamma_1 \ldots \gamma_{n-1}}(\theta',\theta_1  \ldots \theta_{n-1})\nonumber\\
&\qquad \times  \left( \prod_{j=1}^{n-1}S^{\gamma_j \nu}_{\nu \gamma_j}(\theta' -\theta_j +i\theta_{(\kappa \bar\nu)} )\right) \nonumber\\
&\; =\; \sqrt{n} \underset{\kappa\in\I}{\sum_{\nu = \up,\bar\up}} i\eta^{\nu}_{\bar\nu \kappa} \int d\theta'\, g^+_{\bar\nu}(\theta' -i\pi +i\theta_{(\kappa \bar\nu)}) f^+_{\bar\nu} (\theta' +i\theta_{\bar\nu \kappa})
(\Psi_n)^{\kappa \gamma_1 \ldots \gamma_{n-1}}(\theta',\theta_1  \ldots \theta_{n-1})\nonumber\\
&\qquad \times  \left( \prod_{j=1}^{n-1}S^{\gamma_j \nu}_{\nu \gamma_j}(\theta' -\theta_j +i\theta_{(\kappa \bar\nu)} )\right),
\end{align}
where in the second equality we applied the shift $\theta' \rightarrow \theta' +i\theta_{(\kappa \bar\nu)}$
which is legitimate because, when $S^{\gamma_j \nu}_{\nu \gamma_j}$ has no pole there is no problem and
when $S^{\gamma_j \nu}_{\nu \gamma_j}$ has a pole
at $\theta' -\theta_j +i\theta_{(\kappa \bar\nu)}$ then it forces $\gamma_j = \k$ and
the component of $\Psi \in \dom(\fct(f))\cap \dom(\fct'(g))$ has a zero at $\theta' - \theta_j = 0$
(as we explained before this Theorem \ref{theo:commutator} and because $\D_0 \subset \dom(\fct(f))\cap \dom(\fct'(g))$)
which compensates the simple pole of $S^{\gamma_j \nu}_{\nu \gamma_j}$. In the last equality we used the property that $g^-_{\bar\nu}(\zeta)=g^+_{\bar\nu}(\zeta \pm i\pi)$.
More precisely, we are applying the Cauchy theorem to vector-valued functions, which can be
justified by \cite[Lemma B.2]{CT15-1}, \ref{regularity} and \cite[Proposition E.7]{Tanimoto15-2}.

\begin{flushleft} 
 {\it The commutator $[z(J_1f^-), \chi'(g)]$}
 \end{flushleft}
The first term of this commutator gives
\begin{align*}
 & (z(J_1f^-) \chi'(g)\Psi_n)^{\gamma_1 \ldots \gamma_{n-1}}(\theta_1,\cdots,\theta_{n-1}) \\
 &\;=\; \sqrt n \sum_\kappa \int d\theta'\, f^-_{\bar\kappa}(\theta')(\chi'(g) \Psi_n)^{\kappa \gamma_1 \ldots \gamma_{n-1}}(\theta', \theta_1 \ldots \theta_{n-1})\\
  &\;=\; -\sqrt n \sum_{k=1}^{n-1}\sum_{\alpha \beta_k \kappa} i\eta^{\gamma_k}_{\alpha \beta_k}g^+_\alpha\left(\theta_k - i\theta_{(\alpha \beta_k)}\right)
\int d\theta'\, f^-_{\bar\kappa}(\theta')\prod_{k+1\le j \le n-1} S^{\alpha \gamma_j}_{\gamma_j \alpha}\left(\theta_j-\theta_k + i\theta_{(\alpha \beta_k)}\right) \\
&\qquad \times ( \Psi_n)^{\kappa \gamma_1 \ldots \beta_k \ldots \gamma_{n-1}}\left(\theta',\theta_1,\cdots,\theta_k +i\theta_{(\beta_k \alpha)},\cdots, \theta_{n-1}\right) \\
  &\;-\; \sqrt n \sum_{\alpha \beta \kappa}i\eta^{\kappa}_{\alpha \beta}\int d\theta'\, f^-_{\bar\kappa}(\theta') g^+_\alpha\left(\theta' - i\theta_{(\alpha \beta)}\right)\prod_{1\le j \le n-1} S^{\alpha \gamma_j}_{\gamma_j \alpha}\left(\theta_j-\theta'+i \theta_{(\alpha \beta)}\right)\\
&\qquad \times  (\Psi_n)^{\beta \gamma_1 \ldots \gamma_{n-1}}\left(\theta'+i\theta_{(\beta \alpha)},\theta_1,\cdots,\theta_{n-1}\right),
 \end{align*}
where we used \eqref{chip} after renaming of the summation index $k$ into $n-k+1$. 

Similarly, we have
 \begin{align*}
  &( \chi'(g)z(J_1f^-)\Psi_n)^{\gamma_1 \ldots \gamma_{n-1}}(\theta_1,\cdots,\theta_{n-1}) \\
  &\;=\; -\sqrt n \sum_{k=1}^{n-1} \sum_{\alpha \beta_k} i\eta^{\gamma_k}_{\alpha \beta_k}g^+_\alpha\left(\theta_k - i\theta_{(\alpha \beta_k)}\right)
\prod_{k+1\le j \le n-1} S^{\alpha \gamma_j}_{\gamma_j \alpha}\left(\theta_j-\theta_k + i\theta_{(\alpha \beta_k)}\right)  \\
&\qquad \times (z(J_1 f^-) \Psi_n)^{\gamma_1 \ldots \beta_k \ldots \gamma_{n-1}}\left(\theta_1,\cdots,\theta_k +i\theta_{(\beta_k \alpha)},\cdots, \theta_{n-1}\right)\\
  &\;=\; -\sqrt n \sum_{k=1}^{n-1}\sum_{\alpha \beta_k \kappa} i\eta^{\gamma_k}_{\alpha \beta_k}g^+_\alpha\left(\theta_k - i\theta_{(\alpha \beta_k)}\right)
 \prod_{k+1\le j \le n-1} S^{\alpha \gamma_j}_{\gamma_j \alpha}\left(\theta_j-\theta_k + i\theta_{(\alpha \beta_k)}\right)\\
&\qquad \times \int d\theta'\, f^-_{\bar\kappa}(\theta')(\Psi_n)^{\kappa \gamma_1 \ldots \beta_k \ldots \gamma_{n-1}}\left(\theta',\theta_1,\cdots,\theta_k +i\theta_{(\beta_k \alpha)},\cdots, \theta_{n-1}\right).
 \end{align*}
Now we combine the two above expressions.
Using the assumption that $\a = \up,\bar\up$ (which are the nonzero components of $g$),
$\bar\kappa = \up,\bar\up$ (which are the nonzero components of $f$),
and noting that $\eta_{\k\b}^\k = 0$ by \ref{p:table} (because $\eta_{\a\b}^\g = i\sqrt{2\pi |R_{\a\b}^\g|}$),
we again only have to consider the cases where $\a = \up = \bar \kappa$ and $\a = \bar\up = \bar\kappa$
and obtain
\begin{align*}
&([z(J_1f^-), \chi'(g)]\Psi_n)^{\gamma_1 \ldots \gamma_{n-1}}(\theta_1,\cdots,\theta_{n-1}) \\
&\;=\;-\sqrt{n} \underset{\beta \in \I}{\sum_{\nu = \up,\bar\up}} i\eta^{\nu{}}_{\bar\nu \beta} \int d\theta'\, g^+_{\bar\nu}(\theta' -i\theta_{(\bar\nu \beta)}) f^-_{\bar\nu{}}(\theta')
(\Psi_n)^{\beta \gamma_1 \ldots \gamma_{n-1} }(\theta' +i\theta_{(\beta \bar\nu)},\theta_1  \ldots \theta_{n-1}) \\
&\qquad \times  \prod_{j=1}^{n-1}S^{\bar\nu \gamma_j }_{\gamma_j \bar\nu}(\theta_j -\theta' +i\theta_{(\bar\nu \beta)}) \\
&\;=\;-\sqrt{n} \underset{\beta \in \I}{\sum_{\nu = \up,\bar\up}} i\eta^{\nu{}}_{\bar\nu \beta} \int d\theta'\, g^+_{\bar\nu}(\theta' -i\theta_{\bar\nu \beta}) f^-_{\bar\nu{}} (\theta'-i\theta_{(\beta \bar\nu)})
(\Psi_n)^{\beta \gamma_1 \ldots \gamma_{n-1} }(\theta',\theta_1  \ldots \theta_{n-1}) \\
&\qquad \times  \prod_{j=1}^{n-1}S^{\bar\nu \gamma_j }_{\gamma_j \bar\nu}(\theta_j -\theta' +i\theta_{\bar\nu \beta}) \\
&\;=\;-\sqrt{n} \underset{\beta \in \I}{\sum_{\nu = \up,\bar\up}} i\eta^{\nu{}}_{\bar\nu \beta} \int d\theta'\, g^+_{\bar\nu}(\theta' -i\theta_{\bar\nu \beta}) f^+_{\bar\nu{}} (\theta'+i\pi -i\theta_{(\beta \bar\nu)})
(\Psi_n)^{\beta \gamma_1 \ldots \gamma_{n-1} }(\theta',\theta_1  \ldots \theta_{n-1}) \\
&\qquad \times  \prod_{j=1}^{n-1}S^{\bar\nu \gamma_j }_{\gamma_j \bar\nu}(\theta_j -\theta' +i\theta_{\bar\nu \beta}) \\
\end{align*}
where in the second equality we applied the shift $\theta' \rightarrow \theta' -i\theta_{(\beta \bar\nu)}$,
in the third equality we used the property that $f^+_{\bar\nu}(\zeta +i\pi) = f^-_{\bar\nu}(\zeta)$.
In order to carry out the shift in the second equality, we used again the properties of
$\Psi \in (\D_0 \subset) \dom(\fct(f))\cap\dom(\fct'(g))$
explained before Theorem \ref{theo:commutator} and the $L^2$-valued Cauchy theorem \cite[Lemma B.2]{CT15-1}
(note that
$S^{\bar\nu \gamma_j }_{\gamma_j \bar\nu}(\theta_j -\theta' +i\theta_{\bar\nu \beta})
= S^{\gamma_j \nu}_{\nu \gamma_j}(\theta' -\theta_j +i\pi -i\theta_{\bar\nu \beta})
= S^{\gamma_j \nu}_{\nu \gamma_j}(\theta' -\theta_j +i\theta_{(\beta\bar\nu)})$
as we will see in the next paragraph and we can apply the same argument as
in the computation of the commutator $[\chi(f), z'(J_1g^-)]$: the existence of a pole
forces $\beta = \gamma_j$, and it gets canceled by the zero of $\Psi$. If there is no pole,
there is no problem in the shift).

Now we compare this to Eq.~\eqref{chizp}. By \ref{crossing} and \ref{CPTinvariance}
we have $S^{\bar\nu \gamma_j }_{\gamma_j \bar\nu}(\theta_j -\theta' +i\theta_{\bar\nu \beta})
= S^{\gamma_j \nu}_{\nu \gamma_j}(\theta' -\theta_j +i\pi -i\theta_{\bar\nu \beta})$.
As a consequence of \ref{p:angles} and \ref{p:CPT} we also have
$\theta_{(\beta \bar\nu)} = \pi -\theta_{\bar\nu \beta}$ and
$\theta_{\bar\nu \beta} = \theta_{\nu \bar\beta} = \pi -\theta_{(\beta \bar\nu)}$,
so that the $S, f^+, g^+$ factors coincide.

Hence, the commutator $[z(J_1f^-), \chi'(g)]$ coincides up to a sign with Eq.~\eqref{chizp}, and therefore they cancel.
 
\begin{flushleft} 
{\it The commutators $[z^\dagger(f^+), \chi'(g)]$ and $[\chi(f), z'^\dagger(g^+)]$}
\end{flushleft}
These commutators  are the adjoints of the previous commutators,
therefore, they cancel as a weak commutator as a consequence of the above computations.

\begin{flushleft} 
{\it The commutator $[\phi(f), \phi'(g)]$}
\end{flushleft}

This commutator has been computed in \cite{LS14}. In the present case where $S$ is diagonal, the result in \cite{LS14} reduces to the following expression:
\begin{align}\label{LScomm}
&([\phi'(g), \phi(f)]\Psi_n)^{\pmb{\gamma}}(\theta_1,\cdots,\theta_n) \nonumber\\
&\;=\; \sum_{\nu}\int d\theta'\, \left(  g^-_{\bar\nu}(\theta') \left(\prod_{l=1}^n S^{\gamma_l \nu}_{\nu \gamma_l}(\theta' -\theta_l)\right) f^+_\nu(\theta')
- g^+_\nu(\theta') \left(\prod_{l=1}^n  \overline{S^{\gamma_l \nu}_{\nu \gamma_l}(\theta' -\theta_l)}\right) f^-_{\bar\nu}(\theta')\right) \nonumber\\
&\quad\quad\quad\quad\times(\Psi_n)^{\pmb{\gamma}}(\theta_1, \ldots, \theta_n).
\end{align}
We note that the first term in the integrand is equal to the second term except for a shift of $+i\pi$ in $\theta'$.
Compared to \cite[Theorem 3.2]{LS14}, since $S$ has now poles in the physical strip,
we obtain residues when shifting the integration contour. 

Recall that we consider test functions $f,g$ whose only non-zero components correspond
to $\up$ and $\bar\up$.
In this case, the factor $S_{\nu\a}^{\a\nu}$ appearing in the expression of the commutator have at most two simple poles at
$\zeta = i\theta_{\alpha \nu}$ (if $(\a\nu)\to\b$ is a fusion)
and $\zeta=i\theta'_{\nu\bar\a}$ (if $(\nu\bar\a)\to\g$ is a fusion), where $\nu = \up,\bar\up$,
as specified in \ref{polestructure}.
Therefore, by recalling that $R_{\a\b}^\g = 0$ if $(\a\b)\to\g$ is \textit{not} a fusion, we have
\begin{align}\label{comm:phiphip}
&\frac{1}{2\pi i} ([\phi'(g),\phi(f)]\Psi_n)^{\pmb{\gamma}}(\theta_1, \ldots, \theta_n)\nonumber\\
&\; =\; \underset{\k,\k'\in\I}{\sum_{\nu = \up,\bar \up}} \left( \sum_{j=1}^n R_{\nu \gamma_j}^\k g^-_{\bar\nu}(\theta_j +i\theta_{\nu \gamma_j}) f^+_\nu (\theta_j +i\theta_{\nu \gamma_j}) \left( \prod_{\substack{k=1\\ k \neq j}}^n S^{\gamma_k \nu}_{\nu \gamma_k}(\theta_j +i\theta_{\nu \gamma_j}-\theta_k) \right)\right.\nonumber\\
&\; \quad\quad\quad\quad+\; \left.\sum_{j=1}^n R_{\bar\gamma_j\nu}^{\prime\k'} g^-_{\bar\nu}(\theta_j +i\theta'_{\bar\gamma_j\nu}) f^+_\nu (\theta_j +i\theta'_{\bar\gamma_j\nu}) \left( \prod_{\substack{k=1\\ k \neq j}}^n S^{\gamma_k \nu}_{\nu \gamma_k}(\theta_j +i\theta'_{\bar\gamma_j\nu}-\theta_k) \right)\right) \nonumber \\
&\quad\quad\quad\times (\Psi_n)^{\gamma_1 \ldots \gamma_n}(\theta_1, \ldots,\theta_n)\nonumber\\
&\; =\; \underset{\k,\k'\in\I}{\sum_{\nu = \up,\bar \up}} \left( \sum_{j=1}^n R_{\nu \gamma_j}^\k g^+_{\bar\nu}(\theta_j +i\theta_{\nu \gamma_j}-i\pi) f^+_\nu (\theta_j +i\theta_{\nu \gamma_j}) \left( \prod_{\substack{k=1\\ k \neq j}}^n S^{\gamma_k \nu}_{\nu \gamma_k}(\theta_j +i\theta_{\nu \gamma_j}-\theta_k) \right)\right.\nonumber\\
&\; \quad\quad\quad\quad-\; \left.\sum_{j=1}^n  R_{\gamma_j\bar\nu}^{\bar\k'} g^+_{\bar\nu}(\theta_j +i\theta'_{\gamma_j\bar\nu}-i\pi) f^+_\nu (\theta_j +i\theta'_{\gamma_j\bar\nu}) \left( \prod_{\substack{k=1\\ k \neq j}}^n S^{\gamma_k \nu}_{\nu \gamma_k}(\theta_j +i\theta'_{\gamma_j\bar\nu}-\theta_k) \right)\right) \nonumber \\
&\quad\quad\quad \times (\Psi_n)^{\gamma_1 \ldots \gamma_n}(\theta_1, \ldots,\theta_n), \nonumber\\
&\; =\; \underset{\k\in\I}{\sum_{\nu = \up,\bar \up}} \left( \sum_{j=1}^n R_{\nu \gamma_j}^\k g^+_{\bar\nu}(\theta_j +i\theta_{\nu \gamma_j}-i\pi) f^+_\nu (\theta_j +i\theta_{\nu \gamma_j}) \left( \prod_{\substack{k=1\\ k \neq j}}^n S^{\gamma_k \nu}_{\nu \gamma_k}(\theta_j +i\theta_{\nu \gamma_j}-\theta_k) \right)\right.\nonumber\\
&\; \quad\quad\quad\quad-\; \left.\sum_{j=1}^n  R_{\gamma_j\nu}^\k g^+_{\nu}(\theta_j +i\theta'_{\gamma_j\nu}-i\pi) f^+_{\bar\nu} (\theta_j +i\theta'_{\gamma_j\nu}) \left( \prod_{\substack{k=1\\ k \neq j}}^n S^{\gamma_k \bar\nu}_{\bar\nu \gamma_k}(\theta_j +i\theta'_{\gamma_j\nu}-\theta_k) \right)\right) \nonumber \\
&\quad\quad\quad \times (\Psi_n)^{\gamma_1 \ldots \gamma_n}(\theta_1, \ldots,\theta_n),
\end{align}
where in the second equality we used the property that $g^-_{\bar\nu}(\zeta)=g^+_{\bar\nu}(\zeta \pm i\pi)$,
\ref{p:bootstrap}, CPT-symmetry on $R$'s and $\theta$'s and
in the last equality we just changed the dummy indices.

Note that the product of $S$ factors in Eq.~\eqref{LScomm} might raise some concern about the nature of the poles of the integrand and whether the sum of the residues over the simple poles in the first equality of Eq.~\eqref{comm:phiphip} is justified. We will argue as follows: If any pair of $\theta_l$'s does not coincide, then the integrand has two simple poles and the first equality in \eqref{comm:phiphip} follows. The complement of this set of $\theta_l$'s has Lebesgue measure zero, and therefore the equality still holds.

We also remark that the integral expression in \eqref{comm:phiphip} seems unbounded due to the poles of the $S$-factors, but since $\phi(f),\phi'(g)$ are bounded on vectors with finite particle number,
one can expect that these divergences cancel.
Indeed, for each values of $j,k$  the $S$-factors $S^{\gamma_k \nu}_{\nu \gamma_k}(\theta_j +i\theta_{\nu \gamma_j}-\theta_k)$
and $S^{\gamma_k \bar\nu}_{\bar\nu \gamma_k}(\theta_j +i\theta'_{\nu \gamma_j}-\theta_k)$ in the products above have a pole at
$\theta_j -\theta_k =0$ for $\gamma_j = \gamma_k$. Correspondingly, the terms with the values of $j,k$ exchanged have
also a pole at $\theta_j -\theta_k =0$, but with the negative of the previous residues.
The factors $f^+, g^+$ with exchanged variables $\theta_j, \theta_k$
coincide when $\theta_j = \theta_k$ and $\gamma_j = \gamma_k$. This implies that such terms cancel.

\begin{flushleft} 
{\it The commutator $[\chi(f), \chi'(g)]$}
\end{flushleft}

We compute this commutator between vectors $\Psi, \Phi$ with only $n$-particle components.
Recall the expressions of $\chi(f)$ and $\chi'(g)$ at the beginning of Sec.~\ref{chi},
where they are written as the sum of $n$ operators acting on different variables,
therefore, there are $n^2$ terms in each of
the scalar products $\<\chi'(g)\Phi, \chi(f)\Psi\>$ and $\<\chi(f)\Phi, \chi'(g)\Psi\>$.
Of these, the $n(n-1)$ terms in which the above-mentioned operators act on different variables
give exactly the same contribution (as in \cite{CT15-1}), which we denote by $C$, therefore,
they cancel in the commutator and hence are irrelevant.
The relevant part is
\begin{align*}
&\<\chi'(g)\Phi, \chi(f)\Psi\> - C \\
=& \sum_{\pmb{\gamma} \alpha \alpha'} \sum_{k=1}^n \sum_{\beta_k \beta'_{k}} \sqrt{2\pi |R_{\alpha \beta_k}^{\gamma_k}|}  \int d\theta_1 \ldots d\theta_n \left( \prod_{j=1}^{k-1}S^{\gamma_j \alpha}_{\alpha \gamma_j}(\theta_k -\theta_j +i\theta_{(\alpha \beta_k)}) \right)\\
&\; \times \; f^+_{\alpha}(\theta_k +i\theta_{(\alpha \beta_k)}) (\Psi_n)^{\gamma_1 \ldots \beta_k \ldots \gamma_n} (\theta_1, \ldots, \theta_k -i\theta_{(\beta_k \alpha)}, \ldots, \theta_n)\\
&\; \times\; \sqrt{2\pi |R_{\alpha' \beta'_k}^{\gamma_k}|} \overline{ \left( \prod_{p= k+1}^n S^{\alpha' \gamma_p}_{\gamma_p \alpha'}(\theta_p -\theta_k +i\theta_{(\alpha' \beta'_k)}) \right)}\\
&\; \times \; \overline{g^+_{\alpha'}(\theta_k -i\theta_{(\alpha' \beta'_k)})} \overline{(\Phi_n)^{\gamma_1 \ldots \beta'_k \ldots \gamma_n}(\theta_1, \ldots, \theta_k +i\theta_{(\beta'_k \alpha')}, \ldots, \theta_n)}\\
=& \sum_{\pmb{\gamma} \alpha \alpha'} \sum_{k=1}^n \sum_{\beta_k \beta'_{k}} \sqrt{2\pi |R_{\alpha \beta_k}^{\gamma_k}|} \int d\theta_1 \ldots d\theta_n \left( \prod_{j=1}^{k-1}S^{\gamma_j \alpha}_{\alpha \gamma_j}(\theta_k -\theta_j +i\theta_{\alpha \beta_k}) \right)\\
&\; \times \; f^+_{\alpha}(\theta_k +i\theta_{\alpha \beta_k}) (\Psi_n)^{\gamma_1 \ldots \beta_k \ldots \gamma_n} (\theta_1, \ldots, \theta_k, \ldots, \theta_n)\\
&\; \times\; \sqrt{2\pi |R_{\alpha' \beta'_k}^{\gamma_k}|} \left( \prod_{p= k+1}^n S_{\alpha' \gamma_p}^{\gamma_p \alpha'}(\theta_k -\theta_p  +i\theta_{(\beta_k \alpha)}+i\theta_{(\alpha' \beta'_k)}) \right)\\
&\; \times \; g^+_{\bar{\alpha'}}(\theta_k +i\theta_{(\alpha' \beta'_k)} +i\theta_{(\beta_k \alpha)}-i\pi) \overline{(\Phi_n)^{\gamma_1 \ldots \beta'_k \ldots \gamma_n}(\theta_1, \ldots, \theta_k +i\theta_{(\beta'_k \alpha')}-i\theta_{(\beta_k \alpha)}, \ldots, \theta_n)},
\end{align*}
where in the second equality we implemented the shift $\theta_k \rightarrow \theta_k +i\theta_{(\beta_k \alpha)}$,
we used parity symmetry \ref{parity}, unitarity \ref{unitarity}, hermitian analyticity \ref{hermitian}
and the property that $\overline{g^+_{\alpha'}(\zeta)} =g^-_{\bar{\alpha'}}(\bar\zeta) = g^+_{\bar{\alpha'}}(\zeta -i\pi)$. 
To perform the shift in $\theta_k$ we used the analyticity and decay properties of $f^+, g^+$ at infinity in the strip,
\cite[Lemma B.2]{CT15-1} and the property of $\Psi, \Phi \in \D_0$ explained before Theorem \ref{theo:commutator}.
This last property tells that $\Psi, \Phi$ have zeros on the real hyperplane,
so that the product of $\Psi, \Phi$ and the $S$ factors is analytic and bounded in $\theta_k$,
and that the final result, after the shift, remains $L^2$ (as explained before this theorem).

Similarly, using the definitions of $\chi(f)$ and $\chi'(g)$ at the beginning of Section \ref{chi}
we can compute the other term $\<\chi(f) \Phi, \chi'(g)\Psi\>$ in the commutator $[\chi(f),\chi'(g)]$ and obtain:
\begin{align*}
&\<\chi(f)\Phi, \chi'(g)\Psi\> - C\\
=& \sum_{\pmb{\gamma}, \alpha, \alpha'} \sum_{k=1}^n \sum_{\beta_k \beta'_k} \sqrt{2\pi |R_{\alpha \beta_k}^{\gamma_k}|} \int d\theta_1 \ldots d\theta_n \overline{\left( \prod_{j=1}^{k-1}S^{\gamma_j \alpha}_{\alpha \gamma_j}(\theta_k -\theta_j +i\theta_{(\alpha \beta_k)}) \right)}\\
&\; \times \; \overline{f^+_{\alpha}(\theta_k +i\theta_{(\alpha \beta_k)})} \overline{ (\Phi_n)^{\gamma_1 \ldots \beta_k \ldots \gamma_n} (\theta_1, \ldots, \theta_k -i\theta_{(\beta_k \alpha)}, \ldots, \theta_n)}\\
&\; \times\; \sqrt{2\pi |R_{\alpha' \beta'_k}^{\gamma_k}|} \left( \prod_{p= k+1}^n S^{\alpha' \gamma_p}_{\gamma_p \alpha'}(\theta_p -\theta_k +i\theta_{(\alpha' \beta'_k)}) \right)\\
&\; \times \; g^+_{\alpha'}(\theta_k -i\theta_{(\alpha' \beta'_k)})(\Psi_n)^{\gamma_1 \ldots \beta'_k \ldots \gamma_n}(\theta_1, \ldots, \theta_k +i\theta_{(\beta'_k \alpha')}, \ldots, \theta_n)\\
=&\sum_{\pmb{\gamma}, \alpha, \alpha'} \sum_{k=1}^n \sum_{\beta_k \beta'_k} \sqrt{2\pi |R_{\alpha \beta_k}^{\gamma_k}|} \int d\theta_1 \ldots d\theta_n \left( \prod_{j=1}^{k-1}S^{\gamma_j \bar\alpha}_{\bar\alpha \gamma_j}(\theta_k -\theta_j -i\theta_{(\beta'_k \alpha')} -i\theta_{(\alpha \beta_k)} +i\pi) \right)\\
&\; \times \; f^+_{\bar\alpha}(\theta_k -i\theta_{(\beta'_k \alpha')} -i\theta_{(\alpha \beta_k)} +i\pi) \overline{ (\Phi_n)^{\gamma_1 \ldots \beta_k \ldots \gamma_n} (\theta_1, \ldots, \theta_k +i\theta_{(\beta'_k \alpha')}-i\theta_{(\beta_k \alpha)}, \ldots, \theta_n)}\\
&\; \times\; \sqrt{2\pi |R_{\alpha' \beta'_k}^{\gamma_k}|} \left( \prod_{p= k+1}^n S_{\bar{\alpha'} \gamma_p}^{\gamma_p \bar{\alpha'}}(\theta_k -\theta_p -i\theta_{\alpha' \beta'_k} +i\pi) \right)\\
&\; \times \; g^+_{\alpha'}(\theta_k -i\theta_{\alpha' \beta'_k})(\Psi_n)^{\gamma_1 \ldots \beta'_k \ldots \gamma_n}(\theta_1, \ldots, \theta_k, \ldots, \theta_n),
\end{align*}
where in the second equality we implemented the shift $\theta_k \rightarrow \theta_k -i\theta_{(\beta'_k \alpha')}$,
we used parity symmetry \ref{parity}, unitarity \ref{unitarity}, hermitian analyticity \ref{hermitian}, crossing symmetry \ref{crossing}
and the property that $\overline{f^+_{\alpha}(\zeta)} =f^-_{\bar{\alpha}}(\bar\zeta) = f^+_{\bar{\alpha}}(\zeta +i\pi)$.
As before, the shift in $\theta_k$ and the fact that the result is $L^2$ is guaranteed by analyticity and
decay properties of $f^+, g^-$ at infinity in the strip, \cite[Lemma B.2]{CT15-1} and the zeros of the vectors
$\Psi, \Phi \in \D_0$.

We now exploit the properties of elementary particles $\up, \bar\up$.
We first consider the components $\chi_\up$ and $\chi'_\up$.
Since only those indices with processes $(\up\b_k)\to\g_k$ and $(\up\b_k')\to\g_k$
are taken into consideration in the sum, it follows that $\b_k = \b_k'$ (see the assumption on elementary particles
in Section \ref{S-matrix}).
Then the two terms of the commutator above simplify, up to irrelevant terms $C_{\up\up}$ etc.\! as before, as follows:
\begin{align}
&\<\chi'_\up(g)\Phi, \chi_\up(f)\Psi\> - C_{\up\up} \nonumber\\
&= -\sum_{\pmb{\gamma}} \sum_{k=1}^n \sum_{\beta_k} 2\pi i R_{\up \beta_k}^{\gamma_k} \int d\theta_1 \ldots d\theta_n \left( \prod_{j \neq k}S^{\gamma_j \up}_{\up \gamma_j}(\theta_k -\theta_j +i\theta_{\up \beta_k } ) \right) \nonumber\\
&\quad\quad \times \; f^+_{\up}(\theta_k +i\theta_{\up \beta_k} ) \overline{ (\Phi_n)^{\gamma_1 \ldots \beta_k \ldots \gamma_n} (\theta_1, \ldots, \theta_k, \ldots, \theta_n)} g^+_{\bar\up}(\theta_k +i\theta_{\up \beta_k} -i\pi) \nonumber\\
&\quad\quad \times (\Psi_n)^{\gamma_1 \ldots \beta_k \ldots \gamma_n}(\theta_1, \ldots, \theta_k, \ldots, \theta_n), \label{eq:shift1} \\
&\<\chi(f)_\up\Phi, \chi'_\up(g)\Psi\> - C_{\up\up}' \nonumber \\
&= -\sum_{\pmb{\gamma}} \sum_{k=1}^n \sum_{\beta_k} 2\pi i R_{\up \beta_k}^{\gamma_k} \int d\theta_1 \ldots d\theta_n \left( \prod_{j \neq k}S^{\gamma_j \bar\up}_{\bar\up \gamma_j}(\theta_k -\theta_j+i\pi -i\theta_{\up \beta_k}) \right) \nonumber\\
&\quad\quad \times \; f^+_{\bar\up}(\theta_k -i\theta_{\up \beta_k} +i\pi) \overline{ (\Phi_n)^{\gamma_1 \ldots \beta_k \ldots \gamma_n} (\theta_1, \ldots, \theta_k, \ldots, \theta_n)}
g^+_{\up}(\theta_k -i\theta_{\up \beta_k}) \nonumber\\
&\quad\quad \times (\Psi_n)^{\gamma_1 \ldots \beta_k \ldots \gamma_n}(\theta_1, \ldots, \theta_k, \ldots, \theta_n), \label{eq:shift2}
\end{align}
where we used the fact that $|R_{\up \beta}^\gamma| = -iR_{\up \beta}^\gamma$, ``positivity'' of residues.
Analogous expressions can be obtained for $\<\chi'_{\bar\up}(g)\Phi, \chi_{\bar\up}(f)\Psi\>$
and $\<\chi(f)_{\bar\up}\Phi, \chi'_{\bar\up}(g)\Psi\>$, namely, $\up$ should be simply replaced by $\bar\up$.

We compare them with the commutator \eqref{comm:phiphip}.
By $\theta'_{\k \gamma_j} = \pi -\theta_{\k\gamma_j} = \pi- \theta_{\gamma_j\k}$
and exchanging the dummy indices $\b_k\leftrightarrow\g_k$, $j\leftrightarrow k$,
these two commutators cancel each other.

Next, we consider the remaining components.
We first look at the following combination:
\begin{align}
&\<\chi_{\bar\up}'(g)\Phi, \chi_\up(f)\Psi\> - C_{\bar\up\up}\nonumber \\
&=\sum_{\pmb{\gamma}} \sum_{k=1}^n \sum_{\beta_k, \beta'_k} \sqrt{2\pi |R_{\up \beta_k}^{\gamma_k}|}\sqrt{2\pi |R_{\bar\up \beta'_k}^{\gamma_k}|} \nonumber\\
&\quad \times \int d{\pmb \theta}
\left( \prod_{j=1}^{ k-1}S^{\gamma_j \up}_{\up \gamma_j}(\theta_k -\theta_j +i\theta_{\up \beta_k } ) \right)\left( \prod_{p=k+1}^{ n}S^{\gamma_p \bar\up}_{\bar\up \gamma_p}(\theta_k -\theta_p +i\theta_{( \beta_k \up) } +i\theta_{(\bar\up \beta'_k)}) \right) \nonumber\\
&\quad\quad \times \; f^+_{\up}(\theta_k +i\theta_{\up \beta_k} ) g^+_{\up}(\theta_k +i\theta_{(\bar\up \beta'_k)} +i\theta_{(\beta_k \up)} -i\pi) \nonumber\\
&\quad\quad  \times  \overline{ (\Phi_n)^{\gamma_1 \ldots \beta'_k \ldots \gamma_n} (\theta_1, \ldots, \theta_k +i\theta_{(\beta'_k \bar\up)}-i\theta_{(\beta_k \up)}, \ldots, \theta_n)}(\Psi_n)^{\gamma_1 \ldots \beta_k \ldots \gamma_n}(\theta_1, \ldots, \theta_k, \ldots, \theta_n), \label{eq:shift3}\\
&\<\chi_{\bar\up}(f)\Phi, \chi_\up'(g)\Psi\> - C_{\bar\up\up}'\nonumber \\
&=\sum_{\pmb{\gamma}} \sum_{k=1}^n \sum_{\beta_k, \beta'_k}\sqrt{2\pi |R_{\bar\up \beta_k}^{\gamma_k}|} \sqrt{2\pi |R_{\up \beta'_k}^{\gamma_k}|} \nonumber\\
&\quad \times \int d{\pmb \theta}
\left( \prod_{j=1}^{ k-1}S^{\gamma_j \up}_{\up \gamma_j}(\theta_k -\theta_j+i\pi -i\theta_{(\beta'_k \up)}-i\theta_{(\bar\up \beta_k)}) \right) \left( \prod_{p=k+1}^{n}S^{\gamma_p \bar\up}_{\bar\up \gamma_p}(\theta_k -\theta_p+i\pi -i\theta_{\up \beta'_k}) \right) \nonumber\\
&\quad\quad \times \;f^+_{\up}(\theta_k -i\theta_{(\beta'_k \up)}-i\theta_{(\bar\up \beta_k)} +i\pi) g^+_{\up}(\theta_k -i\theta_{\up \beta'_k}) \nonumber\\
&\quad\quad \times  \overline{ (\Phi_n)^{\gamma_1 \ldots \beta_k \ldots \gamma_n} (\theta_1, \ldots, \theta_k +i\theta_{(\beta'_k \up)} -i\theta_{(\beta_k \bar\up)}, \ldots, \theta_n)}(\Psi_n)^{\gamma_1 \ldots \beta'_k \ldots \gamma_n}(\theta_1, \ldots, \theta_k, \ldots, \theta_n). \label{eq:shift4}
\end{align}
By exchanging the dummy indices $\b_k\leftrightarrow\b_k'$ and using the properties of fusion angles for elementary particles (Section \ref{S-matrix}), we will show below that these two terms coincide, which means that they cancel each other in the commutator $[\chi'(g),\chi(f)]$. 

We claim that there hold the following equalities between fusion angles:
\begin{equation*}
\theta_{(\b'_k \up)} = \theta_{(\b_k \bar\up)}, \quad \pi - \theta_{(\b_k \up)} -\theta_{(\bar \up \b'_k)} = \theta_{ \up \b_k }.
\end{equation*}
The first equality follows by the third property of the elementary particles in Sec.~\ref{S-matrix}.  
We prove the second equality as follows. 
Since $R^\g_{\a\b} \neq 0$ only if $(\a\b)\to \g$ is a fusion process, we only have to take into consideration
those indices with $(\b_k \up)\to\g_k$ and $(\bar \up\b'_k)\to\g_k$, which forces $(\gamma_k \upsilon) \to \beta'_k$.
Using the fourth property of the elementary particles in Sec.~\ref{S-matrix}, the second equality becomes
\begin{equation*}
-\theta_{(\b_k \up)} + \theta_{(\up \b'_k)} = \theta_{\up \b_k}.
\end{equation*}
Now, by \ref{p:elempart} we can write  $\beta_k = \upsilon^{\ell}$ for some $\ell$,  and using the first two relations in Eq.~\eqref{relations:elempart}, we find
\begin{equation*}
-\theta_0 + (\ell +2)\theta_0 = (1+ \ell)\theta_0,
\end{equation*}
which concludes the proof of the second equality.

Similarly, we have
\begin{align*}
&\<\chi_{\up}(f)\Phi, \chi_{\bar\up}'(g)\Psi\> - C_{\bar\up\up}'\nonumber \\
&=\sum_{\pmb{\gamma}} \sum_{k=1}^n \sum_{\beta_k, \beta'_k}\sqrt{2\pi |R_{\up \beta_k}^{\gamma_k}|} \sqrt{2\pi |R_{\bar\up \beta'_k}^{\gamma_k}|} \nonumber\\
&\quad \times \int d{\pmb \theta} 
\left( \prod_{j=1}^{ k-1}S^{\gamma_j \bar\up}_{\bar\up \gamma_j}(\theta_k -\theta_j+i\pi -i\theta_{(\beta'_k\bar\up)}-i\theta_{(\up \beta_k)}) \right) \left( \prod_{p=k+1}^{n}S^{\gamma_p \up}_{\up \gamma_p}(\theta_k -\theta_p+i\pi -i\theta_{\bar\up \beta'_k}) \right)\nonumber\\
&\quad\quad \times \;f^+_{\bar\up}(\theta_k -i\theta_{(\beta'_k\bar\up)}-i\theta_{(\up \beta_k)} +i\pi) g^+_{\bar\up}(\theta_k -i\theta_{\bar\up \beta'_k}) \nonumber\\
&\quad\quad \times  \overline{ (\Phi_n)^{\gamma_1 \ldots \beta_k \ldots \gamma_n} (\theta_1, \ldots, \theta_k +i\theta_{(\beta'_k \bar\up)} -i\theta_{(\beta_k \up)}, \ldots, \theta_n)}(\Psi_n)^{\gamma_1 \ldots \beta'_k \ldots \gamma_n}(\theta_1, \ldots, \theta_k, \ldots, \theta_n),\nonumber\\
&\<\chi'_\up(g)\Phi, \chi_{\bar\up}(f)\Psi\> - C_{\bar\up\up} \nonumber \\
&=\sum_{\pmb{\gamma}} \sum_{k=1}^n \sum_{\beta_k, \beta'_k} \sqrt{2\pi |R_{\bar\up \beta_k}^{\gamma_k}|}\sqrt{2\pi |R_{\up \beta'_k}^{\gamma_k}|} \nonumber\\
&\quad \times \int d{\pmb \theta} 
\left( \prod_{j=1}^{ k-1}S^{\gamma_j \bar\up}_{\bar\up \gamma_j}(\theta_k -\theta_j +i\theta_{\bar\up \beta_k } ) \right)\left( \prod_{p=k+1}^{ n}S^{\gamma_p \up}_{\up \gamma_p}(\theta_k -\theta_p +i\theta_{( \beta_k \bar\up) } +i\theta_{(\up \beta'_k)}) \right) \nonumber\\
&\quad\quad \times \; f^+_{\bar\up}(\theta_k +i\theta_{\bar\up \beta_k} ) g^+_{\bar\up}(\theta_k +i\theta_{(\up \beta'_k)} +i\theta_{(\beta_k \bar\up)} -i\pi) \nonumber\\
&\quad\quad  \times  \overline{ (\Phi_n)^{\gamma_1 \ldots \beta'_k \ldots \gamma_n} (\theta_1, \ldots, \theta_k +i\theta_{(\beta'_k \up)}-i\theta_{(\beta_k \bar\up)}, \ldots, \theta_n)}(\Psi_n)^{\gamma_1 \ldots \beta_k \ldots \gamma_n}(\theta_1, \ldots, \theta_k, \ldots, \theta_n),
\end{align*}
and these terms cancel each other by the same properties of fusion angles for elementary particles.

\end{proof}
We would like to remark that,
the restriction to the elementary particles
is crucial for the vanishing of the commutator $[\fct(f), \fct'(g)]$.
Indeed, in examples where we drop such a restriction one can see that the commutator $[\phi(f), \phi'(g)]$ contains contributions from
double or higher poles, which do not vanish.

In the literature on the form factor program, the commutation relations between local fields
have been claimed for $Z(N)$-Ising model with some specific vectors \cite[Appendix D]{BFK06},
but only a short comment on the general case is given.

\subsubsection*{A larger domain of weak commutativity}
The domain $\D_0$ is properly included in $\dom(\fct(f))\cap\dom(\fct'(g))$.
We point out here that for models with two species of particles, for example
the $Z(3)$-Ising model and the $A_2$-affine Toda field theory, this restriction is unnecessary.

\begin{theorem}\label{th:commutator2}
We consider that class of examples with only two species of particles as in Section \ref{examples}.
Let $f$ and $g$ be test functions supported in $W_\L$ and $W_\R$, respectively, and with the property that $f=f^*$ and $g=g^*$.
Furthermore, assume that $f,g$ have components $f_\alpha =0$ and $g_\alpha =0$ for all $\alpha \in \I$
except the indices $\up,\bar\up$ corresponding to the elementary particles.

Then, for each $\Phi, \Psi$ in $\dom(\fct(f))\cap\dom(\fct'(g))$, we have
 \[
 \<\fct(f)\Phi, \fct'(g)\Psi\> = \<\fct'(g)\Phi, \fct(f)\Psi\>.
 \]
\end{theorem}
\begin{proof}
 Most of the proof of Theorem \ref{theo:commutator} works for $\dom(\fct(f))\cap\dom(\fct'(g))$
 and the only points where we used the properties of $\D_0$ were 
 the shifts of integral contours of the terms
 \eqref{eq:shift1}, \eqref{eq:shift2}, \eqref{eq:shift3} and \eqref{eq:shift4}.

 As for \eqref{eq:shift1}, we have $\up = 1$ and the problem with the shift occurs
 only if $R^{\g_k}_{\up\b_k} \neq 0$, hence $\b_k = 1$ and the $S^{\g_j 1}_{1\g_j}$-factors have poles
 at $\theta_{\up\beta_k} = \theta_{11} = \frac{2\pi}3$,
 which happens only if $\g_j = 1$ ($\g_j = 3$ does not exist). But in this case, $\Psi$ has zeros which compensate the poles
 and the shifts are legitimate as we saw before Theorem \ref{theo:commutator}.

 An analogous argument works for \eqref{eq:shift2}, \eqref{eq:shift3} and \eqref{eq:shift4}.
\end{proof}

On the other hand, we show in Appendix \ref{failure} that
our argument for the weak commutativity on $\dom(\fct(f))\cap\dom(\fct'(g))$ fails for the $Z(4)$-Ising model
or any other models with more than two species of particles.

\subsubsection*{Reeh-Schlieder property}
In the models we consider, there are many ``composite'' particles, but
we can construct weakly commuting polarization-free generators only for those particles
which correspond to elementary particles $\up,\bar\up$.
Therefore, it is important to know whether these operators generate the whole Hilbert space.
We do expect this, yet we can give a complete argument only for models with two species of particles.
These include the $Z(3)$-Ising model and the $A_2$-affine Toda field theory.
Let us give a proof for the particular cases and sketch how far we can go in general.

In the operator-algebraic approach, we are interested in the following question.
Let us suppose that for each $g$ such that $g_\a = 0$ for $\a \neq \up,\bar\up$,
there is a self-adjoint extension of $\fct'(g)$, which we denote by the same symbol,
such that $\fct'$ is covariant with respect to $U$.
Suppose also that, for each $f$ such that $f_\a = 0, \a \neq \up,\bar\up$,
$\fct(f)$ has a nice self-adjoint extension, such that
$\fct'(g)$ and $\fct(f)$ \textit{strongly commute} (which we have not proved yet).
We consider the von Neumann algebra
\[
 \M = \{e^{i\fct'(g)}: g_\a = 0 \text{ for } \a \neq \up, \bar\up \text{ and }\supp g \subset W_\R,\},
\]
and we wish to show that $\overline{\M\Omega} = \H$.
Actually, as $\M$ is an algebra of bounded operators containing the identity operator $\1$,
we can freely use the fact that $\overline{\M\Omega} = \M\overline{\M\Omega}$.
Furthermore, by the assumed covariance it holds that $\Ad U(a,0)(\M) \subset \M$
and $U(a,0)\M\Omega \subset \M\Omega$ for $a\in W_\R$.

By the standard Reeh-Schlieder argument, it suffices to show that $\{U(a,0)\M\Omega: a \in \RR^2\}$ is
total in $\H$. Indeed, let us take $\Psi$ from the orthogonal complement of $\M\Omega$.
Then, for any element $x\in\M$, consider the function
$f_{\Psi,x}(a) = \<U(a,0)x\Omega, \Psi\>$ in $\RR^2$.
By the spectral condition for $U$, it has an analytic continuation to $\RR^2 + iV_+$, which is
equal to $0$ for $a \in W_\R$. Now by the edge of the wedge theorem, this function
can be further continued to a complex open region including $\RR^2$, which is $0$ on the two-dimensional real open set
$W_\R$, therefore, must be $0$. This implies that $\Psi$ is orthogonal to any vector
$U(a,0)x\Omega$ for arbitrary $a \in \RR^2, x\in\M$, therefore, must be zero by assumption.

Note that this allows us to consider smearing $U(f)x\Omega := \int d^2a\,f(a)U(a,0)x\Omega$
with a test function $f$. $U$ is strongly continuous,
hence $U(f)x\Omega \in \overline{\lspan\{U(a,0)\M\Omega: a \in \RR^2\}}$.
As $U(f)$ is a bounded operator, we only have to show that
\[
\overline{\lspan\{\left(U(f) + U(a,0)\right)\overline{\M\Omega}: a \in \RR^2, \supp f \subset \RR^2\}} = \H.
\]

Take first the one-particle space. For $\nu = \up,\bar\up$ we have $\fct'(g)\Omega \in \H_1$
and this is in the above closure because $i\fct'(g)\Omega = \frac{d}{dt}e^{it\fct'(g)}\Omega$ and
$\Omega$ is in the domain of $\fct'(g)$ \cite[Theorem VIII.7]{RSII}.
The reality condition on $g$ is that $g_\a(a) = \overline{g_{\overline \a}(a)}$.
For a real scalar-valued function $h$, the pair of functions
$g_{1,\up}(a) = h(a), g_{1,\bar\up}(a) = h(a)$ satisfies the
reality condition, as well as another pair $g_{2,\up(a)} = ih(a), g_{2,\bar\up}(a) = -ih(a)$. By taking the complex linear hull, $h^+ \in \H_{1,\up}$
and $h^+ \in \H_{1,\bar\up}$ are contained in $\overline{\M\Omega}$.
By the one-particle Reeh-Schlieder property (e.g. \cite[Theorem 3.2.1]{Longo08}),
it follows that $(\CC\Omega \oplus \H_{1,\up} \oplus \H_{1,\bar\up}) \subset \overline{\M\Omega}$.

We exhibit how other species of particles can be obtained.
Note that $\H_{1,\up} \cap \dom(\fct'(g))$, which includes $\H_{1,\up} \cap \dom(\chi'(g))$,
is a dense subspace. For any such vector $\xi \in \H_{1,\up} \cap \dom(\chi'(g))$,
$\fct'(g)\xi = \frac{d}{dt}e^{it\fct'(g)}\xi \in (\CC\Omega \oplus \H_{1,\k} \oplus \H_2)$
where $(\up\up)\to\k$.
Note that the joint spectrum of the vector $\fct'(g)\xi$ with respect to $U$
has disjoint three components: the point $(0,0)$, the mass hyperboloid of the single particle $\k$,
and the two-particle spectrum. Let $f$ be a test function whose Fourier transform
has a support in a neighborhood of the mass hyperboloid of the particle $\k$.
For such $f$, the spectrum of $U(f)\fct'(g)\xi$ is concentrated on the mass hyperboloid,
namely, we obtain the one-particle state of the particle $\k$.
As the range of $\chi'_1(g)$ is dense, we obtain $\H_{1,\k} \subset \overline{\M\Omega}$.
Next, note that $(\k\bar\up)\to\up$, and we already have $\H_{1,\up} \subset \overline{\M\Omega}$.
By repeating this way, we obtain the whole one-particle space $\H_1$,
as we assumed that all single-particle states are composite, namely,
can be obtained by fusing $\up$ repeatedly.

The rest is shown by induction. We can complete this passage only for
models with two species of particles.
Suppose that $\H_n \subset \overline{\M\Omega}$.
As above, one can find vectors $\Psi \in \H_n \cap \dom(\fct'(g))$.
Then by applying $\fct'(g)$ with various $g$ as above,
we see $P_{n+1}(\H_n\otimes \H_{1,\up}) \subset \overline{\M\Omega}$.
Together with $\H_{1,\bar\up},$ this completes the induction
for models with two species of particles.
Therefore, if we manage to show that $\fct(f)$ and $\fct'(g)$ commute strongly
(on an appropriate domain), we can obtain a Borchers triple $(\M,U,\Omega)$.

Let us explain what is missing for the general case.
Again, one can find vectors $\Psi \in \dom(\fct'(g)) \cap P_{n+1}(\H_n\otimes \H_{1,\up}) \subset \overline{\M\Omega}$
and separate the $P_{n+1}(\H_n\otimes\H_{1,\k})$-component of $\fct'(g)\Psi$
by smearing with test functions, but the energy-momentum spectrum of this vector
is restricted below a certain mass shell and we do not know how to fill the rest.
Formally, $\phi'(g)$ and $\chi'(g)$ have different momentum transfer, therefore,
it should be possible to separate them by smearing, but if the smearing function
has a non-compact support, there is no guarantee that there is a dense domain.
Therefore, a more refined analysis is needed.

\subsubsection*{(Non-)temperateness of the polarization-free generators}
It is not difficult to show that $\fct(f)$ is not temperate by showing that for any vector $\Psi$
in $\dom(\fct(f))$, $\dom(\fct(f))\Psi$ cannot be polynomially bounded
as in \cite[end of Section 3.3]{CT15-1}.
One can also formally derive the expression for a polarization-free generator $\fct(f)$
by assuming various domain properties of
form factors as in \cite[Section 4.2]{CT15-1} and the existence of a Haag-Kastler net \cite{BBS01}.
It does not in principle exclude the existence of a temperate PFG, but such an operator
must have either a much subtler domain that would invalidate most of the formal computations
or a complicated expansion in terms of $z^\dagger, z$ so that the advantage of the wedge-local
approach would be ruined (a simple expression for PFGs is also necessary for
the proofs of Bisognano-Wichmann property, c.f\! \cite{BL04} and for the modular nuclearity, c.f.\! \cite{Lechner08, Alazzawi14}).
Therefore, we expect that our $\fct(f), \fct'(g)$ should be the right polarization-free generators
for the Haag-Kastler nets for integrable models with bound states.

\section{Conclusions and outlook}\label{sec:conclusions}

Wedge-local observables play an important role in the operator-algebraic construction of a QFT.
In this work we have extended the results in \cite{CT15-1} to models with a factorizing S-matrix with
poles in the physical strip and with a richer particle spectrum.
In particular, we constructed weakly commuting candidate operators for wedge-local observables for the $Z(N)$-Ising model and the
$A_{N-1}$-Toda field theories, which are examples of
a model with $N-1$ species of particles and of several bound states (depending on $N$)
where the S-matrix is of ``diagonal'' form.
Here, candidates for wedge-local observables are obtained as a multi-particle generalization of the operators introduced in \cite{CT15-1}
by the addition of the ``bound state'' operator to the fields of Lechner-Sch\"utzenhofer \cite{LS14}
and by restricting to components corresponding only to particles of species $1$ and $N-1$. 

Moreover, by assuming the existence of nice self-adjoint extensions,
we can show the Reeh-Schlieder property for models with two species of particles.
We also saw that the domain of weak commutativity can be larger for this case.
They include the $Z(3)$-Ising model and the $A_2$-affine Toda field theory.
For models with genuine bound states, we need a more sophisticated method.

A major open problem is to prove that such nice self-adjoint extensions exist and to show the strong commutativity
of $\fct(f)$ and $\fct'(g)$. Some partial results in this direction are currently available only for
scalar S-matrices with bound states (e.g. the Bullough-Dodd model) in \cite{Tanimoto15-1}.
This is a natural first step towards the construction of the corresponding wedge-algebras and to prove
the existence of strictly local observables through intersection of a shifted right and left wedge.
For the scalar S-matrices, the existence of self-adjoint extension is immediate for some choice of $f$.
The situation will be different for general S-matrices, as the proof at hand for scalar S-matrices relies on
certain positivity argument for $\chi(f)$.
Besides, for S-matrices without poles, the models can be realized as the deformations of the free field \cite{Lechner12}.
This point of view seems difficult to maintain for S-matrices with poles, at least by the same techniques,
since our candidates for wedge-local observables are unbounded on each subspace with the fixed particle number,
while the deformed fields of Lechner are bounded there.
As a further step, a proof of modular nuclearity must be also substantially modified from that of the scalar S-matrices,
as it exploits again the positivity of $\chi(f)$ \cite{CT15-3}.

Another interesting problem would be the extension of our construction to integrable models
with ``non-diagonal'' S-matrices, e.g., the Thirring model. We expect that weak wedge-commutativity holds at least for certain components of the PFGs, for example when restricted to
a truncated version of this model where we allow the so-called ``soliton'', ``anti-soliton'' and only one ``breather''. Also, it holds in the case
where there are only ``breathers'' (since this subpart of the S-matrix is in fact ``diagonal''), referred to as the sine-Gordon model.
We will present the last named result elsewhere \cite{CT-sine}.
More general affine Toda field theories and their quantum group symmetry should be also investigated.

\subsubsection*{Acknowledgement}
We thank Bert Schroer for interesting comments, and Karim Shedid for pointing out a typo.
D.C.\! thanks Henning Bostelmann for helpful discussions.
Y.T.\! is supported by Grant-in-Aid for JSPS fellows 25-205.

\appendix
\section{Failure of weak commutativity on the intersection domain}\label{failure}
Here we show that Theorem \ref{theo:commutator} does not hold for $\dom(\fct(g))\cap\dom(\fct'(g))$
if we consider the $Z(4)$-Ising model.

Let us take $\Psi, \Phi \in \H_2$ with only nonzero components $\Psi^{1,3}$ and $\Psi^{3,1}$ (we use the comma to separate indices so not to confuse them with two digit numbers).
Recall that $S$-symmetry means $\Psi^{1,3}(\theta_1,\theta_2) = S_{3,1}^{1,3}(\theta_2-\theta_1)\Psi^{3,1}(\theta_2,\theta_1)$.
The condition $\Psi\in\dom(\fct(f))\cap\dom(\fct'(g))$ only says that
their components should have certain analytic continuations, therefore,
one can choose one of the components, say $\Psi^{1,3}$, quite arbitrarily.
Especially, it generically does not have a zero at $\theta_2 - \theta_1 = 0$.

We take symmetric multi-component test functions $f,g$ with nonzero components with $\up,\bar\up$.
It is straightforward to calculate:
\begin{align*}
 2(\chi_1(f_1)\Psi)^{2,3}(\pmb \theta) &= \sqrt{2\pi |R_{1,1}^2|}f_1^+(\theta_1 + i\theta_{(1,1)})\Psi^{1,3}(\theta_1-i\theta_{(1,1)},\theta_2), \\
 2(\chi_1(f_1)\Psi)^{3,2}(\pmb \theta) &= \sqrt{2\pi |R_{1,1}^2|}S_{1,3}^{3,1}(\theta_2-\theta_1+i\theta_{(1,1)})f_1^+(\theta_2 + i\theta_{(1,1)})\Psi^{3,1}(\theta_1, \theta_2-i\theta_{(1,1)}), \\
 2(\chi_3(f_3)\Psi)^{2,1}(\pmb \theta) &= \sqrt{2\pi |R_{3,3}^2|}f_3^+(\theta_1 + i\theta_{(3,3)})\Psi^{3,1}(\theta_1-i\theta_{(3,3)},\theta_2), \\
 2(\chi_3(f_3)\Psi)^{1,2}(\pmb \theta) &= \sqrt{2\pi |R_{3,3}^2|}S_{1,3}^{3,1}(\theta_2-\theta_1+i\theta_{(3,3)})f_3^+(\theta_2 + i\theta_{(3,3)})\Psi^{1,3}(\theta_1, \theta_2-i\theta_{(3,3)}), \\
 2(\chi'_1(g_1)\Phi)^{3,2}(\pmb \theta) &= \sqrt{2\pi |R_{1,1}^2|}g_1^+(\theta_2 - i\theta_{(1,1)})\Phi^{3,1}(\theta_1, \theta_2+i\theta_{(1,1)}), \\
 2(\chi'_1(g_1)\Phi)^{2,3}(\pmb \theta) &= \sqrt{2\pi |R_{1,1}^2|}S_{3,1}^{1,3}(\theta_2-\theta_1+i\theta_{(1,1)})g_1^+(\theta_1 - i\theta_{(1,1)})\Phi^{1,3}(\theta_1+i\theta_{(1,1)},\theta_2), \\
 2(\chi'_3(g_3)\Phi)^{1,2}(\pmb \theta) &= \sqrt{2\pi |R_{3,3}^2|}g_3^+(\theta_2 - i\theta_{(3,3)})\Phi^{1,3}(\theta_1,\theta_2+i\theta_{(3,3)}), \\
 2(\chi'_3(g_3)\Phi)^{2,1}(\pmb \theta) &= \sqrt{2\pi |R_{3,3}^2|}S_{3,1}^{1,3}(\theta_2-\theta_1+i\theta_{(3,3)})g_3^+(\theta_1 - i\theta_{(3,3)})\Phi^{3,1}(\theta_1+i\theta_{(3,3)},\theta_2). \\
\end{align*}
By recalling that $R_{1,1}^2 = R_{3,3}^2 =: R$ and $\theta_{(1,1)} = \theta_{(3,3)} =: \l$, we obtain
\begin{align*}
 &\frac1{\pi R}\<\chi'(g)\Phi,\chi(f)\Psi\> \\
 &= \int d\theta_1d\theta_2\;\overline{g_1^+(\theta_2 - i\l)\Phi^{3,1}(\theta_1, \theta_2+i\l)} \\
 &\quad\quad\quad\quad \times S_{1,3}^{3,1}(\theta_2-\theta_1+i\l)f_1^+(\theta_2 + i\l)\Psi^{3,1}(\theta_1, \theta_2-i\l) \\
 &\quad+ \int d\theta_1d\theta_2\;\overline{g_3^+(\theta_2 - i\l)\Phi^{1,3}(\theta_1, \theta_2+i\l)} \\
 &\quad\quad\quad\quad \times S_{1,3}^{3,1}(\theta_2-\theta_1+i\l)f_3^+(\theta_2 + i\l)\Psi^{1,3}(\theta_1, \theta_2-i\l) \\
 &= \int d\theta_1d\theta_2\;S_{1,3}^{3,1}(\theta_2-\theta_1+i\l) \\
 &\quad\quad\quad\quad \times \left(\begin{array}{c}
                                    \overline{g_1^+(\theta_2 - i\l)\Phi^{3,1}(\theta_1, \theta_2+i\l)}f_1^+(\theta_2 + i\l)\Psi^{3,1}(\theta_1, \theta_2-i\l) \\
                                    +\overline{g_3^+(\theta_2 - i\l)\Phi^{1,3}(\theta_1, \theta_2+i\l)}f_3^+(\theta_2 + i\l)\Psi^{1,3}(\theta_1, \theta_2-i\l)
                                    \end{array}
                               \right).
\end{align*}
From $S$-symmetry, we can only infer that $\Psi^{1,3}(\theta,\theta) = \Psi^{3,1}(\theta,\theta)$
as $S^{1,3}_{3,1}(0) = 1$, and therefore, the integrand above has generically a pole at $\theta_2-\theta_1+2\l i = 0$
coming from the $S^{1,3}_{3,1}$-factor, which invalidates the application of the Cauchy theorem.
Namely, there are these additional terms (the residues) in the weak commutator
$\<\chi'(g)\Phi,\chi(f)\Psi\> - \<\chi(f)\Phi,\chi'(g)\Psi\>$, which do not vanish for generic
$f,g,\Psi$.

This, however, does not exclude the possibility that $\fct(f)$ and $\fct'(g)$ strongly commute
on a better domain.

{\small
\newcommand{\etalchar}[1]{$^{#1}$}
\def\cprime{$'$} \def\polhk#1{\setbox0=\hbox{#1}{\ooalign{\hidewidth
  \lower1.5ex\hbox{`}\hidewidth\crcr\unhbox0}}}

}

\end{document}